\documentclass[12pt,a4paper]{article}
\usepackage{amsmath,amssymb,amsthm}
\usepackage{microtype}
\usepackage{multirow}
\usepackage{hyperref}

\usepackage{comment}
\usepackage{geometry}

\usepackage{enumerate}

\usepackage{xspace}
\usepackage{scrextend}
\usepackage{mathrsfs}

\usepackage{listings}
\usepackage{color}
\usepackage{cite}
\usepackage{complexity}

\usepackage{graphicx}
\RequirePackage{fancyhdr}
 \usepackage{xcolor}
\usepackage{boxedminipage}

\title{Structural Parameterizations with Modulator Oblivion}
\author{Ashwin Jacob, Venkatesh Raman, Vibha Sahlot \\
The Institute of Mathematical Sciences, HBNI, Chennai, India.\\ \\
Fahad Panolan \\
Department of Computer Science and Engineering, IIT Hyderabad, India. \\
}
\newcommand{\defparproblem}[4]{
 \vspace{1mm}
\noindent\fbox{
  \begin{minipage}{.95\textwidth}
  \begin{tabular*}{\textwidth}{@{\extracolsep{\fill}}lr} \textsc{#1}    \\ \end{tabular*}
  {\bf{Input:}} #2  \\
  {\bf{Parameter:}} #3 \\
  {\bf{Question:}} #4
  \end{minipage}
  }
  \vspace{1mm}
}

\newtheorem{lemma}{Lemma}
\newtheorem{corollary}{Corollary}
\newtheorem{reduction rule}{Reduction Rule}%[subsection]
\newtheorem{branching rule}{Branching Rule}
\newtheorem{definition}{Definition}

\newtheorem{theorem}{Theorem}
\newtheorem{claim}{Claim}
\newtheorem{proposition}{Proposition}

\newcommand{\tw}{{\sf tw}}
\newcommand{\AAA}{{\mathcal A}}

\newcommand{\OO}{{\mathcal O}}

\newcommand{\FF}{{\mathcal F}}

\newcommand{\VCCVD}{{\sc Vertex Cover By CVD}}

\newtheoremstyle{TheoremNum}
        {\topsep}{\topsep}              %%% space between body and thm
        {\itshape}                      %%% Thm body font
        {}                              %%% Indent amount (empty = no indent)
        {\bfseries}                     %%% Thm head font
        {.}                             %%% Punctuation after thm head
        { }                             %%% Space after thm head
        {\thmname{#1}\thmnote{ \bfseries #3}}%%% Thm head spec
    \theoremstyle{TheoremNum}
\newtheorem{thmn}{Theorem}

\begin{document}

\maketitle

\begin{abstract}
%!TEX root = main_lipics.tex

It is known that problems like {\sc Vertex Cover}, {\sc Feedback Vertex Set} and {\sc Odd Cycle Transversal} are polynomial time solvable in the class of chordal graphs. We consider these problems in a graph that has at most $k$ vertices whose deletion results in a chordal graph, when parameterized by $k$. While this investigation fits naturally into the recent trend of what are called `structural parameterizations', here we assume that the deletion set is not given.

One method to solve them is to compute a $k$-sized or an approximate ($f(k)$ sized, for a function $f$)
 chordal vertex deletion set and then use the structural properties of the graph to design an algorithm. 
This method leads to at least $k^{\OO(k)}n^{\OO(1)}$ running time when we use the known parameterized or approximation algorithms for finding a $k$-sized chordal deletion set on an $n$ vertex graph.

In this work, we design $2^{\OO(k)}n^{\OO(1)}$ time algorithms for these problems. Our algorithms do not compute a chordal vertex deletion set (or even an approximate solution). 
Instead, we construct a tree decomposition of the given graph in time $2^{\OO(k)}n^{\OO(1)}$ where each bag is a union of four cliques and $\OO(k)$ vertices. We then apply standard dynamic programming algorithms over this special tree decomposition. This special tree decomposition can be of independent interest. 

Our algorithms are adaptive (robust) in the sense that given an integer $k$, they detect whether the graph has a chordal vertex deletion set of size at most $k$ or output the special tree decomposition and solve the problem. This is analogous to the polynomial algorithm of Raghavan and Spinrad [J. of Algorithms, 2003] for finding a maximum clique in a unit disk graph without the unit disk representation. The algorithm either found a maximum clique in the graph or output a certificate that the given graph was not a unit disk graph, though it was known that determining whether a given graph was unit disk was $NP$-hard.
 
We also show lower bounds for the problems we deal with under the Strong Exponential Time Hypothesis (SETH).

 \end{abstract}
 
%!TEX root = main_lipics.tex

\section{Introduction and Motivation}

Main motivation for parameterized complexity and algorithms is that hard problems have a number of parameters in their input, and feasible algorithms can be obtained when some of these parameters tend to be small. However, barring width parameters (like treewidth and cliquewidth), early parameterizations of problems were mostly in terms of solution size. However starting from the work of Fellows et al~\cite{fellows2007complexity} and Jansen et al~\cite{jansen2013power,fellows2013towards}, the focus shifted to parameterizations by some structure of the input. The motivations for these parameterizations are that 
many problems are computationally easy on special classes of graphs like edge-less graphs, forests and interval graphs. 
Thus parameterizing by the size of a modulator (set of vertices in the graph whose removal results in the easy graph class) became a natural choice of investigation.
%(Bart's thesis) 
 %Fellows et al.\cite{jansen2013power, fellows2013towards}. 
%See also~\cite{bourgeois2013playing, fellows2007complexity}.
% (parameter ecology paper))
 % Explorations of hard problems in input classes that are not `far from' easy (polynomially solvable instances), became a
% source of investigation for some natural definitions of `distance' from easy classes. 
Examples  of such parameterizations include {\sc Clique} and {\sc Feedback Vertex Set} parameterized by the size of minimum vertex cover (i.e., modulator to edge-less graphs),  {\sc Vertex Cover} parameterized by the size of minimum feedback vertex set (i.e., modulator to forests)~\cite{jansen2013power, jansen2013vertex}. See also \cite{DBLP:journals/algorithmica/MajumdarR18, DBLP:journals/mst/Majumdar0018} for more such parameterizations. %\todo{cite PhD thesis of Bart jansen, and more}. 

We continue this line of work on problems in input graphs that are not far from a chordal graph. By distance to a chordal graph, we mean the number of vertices in the graph whose deletion
results in a chordal graph. We call this set as a chordal vertex deletion set (CVD). Specifically, we look at {\sc Vertex Cover}, {\sc Feedback Vertex Set} and {\sc Odd Cycle Transversal} parameterized by the size of a CVD, as these problems are polynomial time solvable in chordal graphs~\cite{golumbic2004algorithmic,spinrad2003efficient,corneil1988complexity}.

In problems for which the parameter is the size of a modulator, it is also assumed that the modulator is given with the input. This assumption can be removed if finding the modulator is also fixed-parameter tractable (FPT) parameterized by the modulator size. However, there are instances where finding the modulator is more expensive than solving the problem if the modulator is given. For example, 
finding a subset of $k$ vertices whose deletion results in a perfect graph is known to be $W$-hard~\cite{DBLP:journals/tcs/HeggernesHJKV13}, whereas if the deletion set is given, then one can show (as explained a bit later in this section) that {\sc Vertex Cover} (thus {\sc Independent Set}) is FPT when parameterized by the size of the deletion set.

Hence Fellows et al.~\cite{FELLOWS2013541} ask whether the {\sc Independent Set} (or equivalently, {\sc Vertex Cover}) is FPT when parameterized by a (promised) bound on the vertex-deletion distance to a 
perfect graph, without giving the deletion set in the input.
While we don't answer this question, we address a similar question in the context of problems parameterized by deletion distance to chordal graphs, another well-studied class of graphs where { \sc Vertex Cover} is polynomial time solvable whereas the best-known algorithm to find a $k$-sized chordal deletion set takes $O^*(k^{O(k)})$ time. In a similar vein to the question by Fellows et al., we ask whether (minimum) {\sc  Vertex Cover} can be solved in $O^*(2^{O(k)})$ time with only a promise on the size $k$ of the chordal deletion set, and answer the question affirmatively.

%\subsection{Our Results}
\smallskip
\noindent
\textbf{Our Results:}
Specifically we give $2^{O(k)}$ algorithms for the following problems.

\defparproblem{\VCCVD}{A graph $G=(V,E)$ and $k, \ell \in \mathbb{N}$.}{Size $k$ of chordal vertex deletion set in $G$.}{
Is there a vertex cover $C$ of size $\ell$ in $G$?
%Does there exist a set $S \subseteq V(G)$ of at most $\ell$ vertices such that $G[V-S]$ is edgeless graph?
}

\defparproblem{{\sc Feedback Vertex Set by CVD} ({\sc FVS by CVD})}{A graph $G=(V,E)$ and $k,\ell \in \mathbb{N}$.}{Size $k$ of chordal vertex deletion set in $G$.}{
Is there a subset $X$ of size at most $\ell$ in $G$ such that $G - X$ is a forest?
%Does there exist a set $S \subseteq V(G)$ of at most $\ell$ vertices such that $G[V-S]$ is edgeless graph?
}

\defparproblem{{\sc Odd Cycle Transversal by CVD} ({\sc OCT by CVD})}{A graph $G=(V,E)$ and $k, \ell \in \mathbb{N}$.}{Size $k$ of chordal vertex deletion set in $G$.}{
Is there a vertex set $X$ of size at most  $\ell$ in $G$ such that $G - X$ is bipartite?
}

We also show that all the problems mentioned above cannot be solved in $\OO^*((2- \epsilon)^{k})$ time  under Strong Exponential Time Hypothesis (SETH) even if a CVD of size $k$ is given as part of the input. 
This matches the upper bound of the known algorithm for \VCCVD\ when the modulator is given.
% Every chordal graph has a clique-tree decomposition, i.e., a tree decomposition where every bag is a clique in the graph. 
%This can be used to obtain a polynomial time algorithms {\sc Vertex Cover} and {\sc Feedback vertex Set} on chordal graphs. 
% Notice that for a graph $G$ if there is a chordal vertex deletion set $S$ of size $k$, then there is a tree decomposition where each bag can be partitioned into $C\uplus N$, where $C$ is a clique and $\vert N\vert \leq k$. This can be proved easily by adding all the vertices of $S$ to every bag of the clique tree decomposition of $G-S$. %Now we add all the vertices in $S$ to every bag. The resulting tree decomposition has the required property. Thus, for any graph that has a chordal vertex deletion set of size at most $k$, there exists such a tree decomposition. %If we are given such a tree decomposition, then we can design DP algorithms for  {\sc Vertex Cover} and {\sc Feedback vertex Set} in $2^{\OO(k)}(n+m)$ time, where $m$ is the number of edge in the graph.  

%\subsection{Related Work}
\smallskip
\noindent
\textbf{Related Work:}
%We obtain the lower bound using a reduction from {\sc Hitting Set} parameterized by the universe size which has a $\Omega^*((2- \epsilon)^{n})$ lower bound under SETH \cite{DBLP:journals/talg/CyganDLMNOPSW16}.
If we are given a CVD $S$ of size $k$ along with an $n$-vertex graph $G$ as the input, then one can easily get a $2^{k}n^{\OO(1)}$ time algorithm (call it $\AAA$) for {\sc Vertex Cover} as follows. First, we guess the subset $X$ of $S$ that is part of our solution. Let $Y$ be the subset of vertices in $V(G)\setminus S$ such that for each $y\in Y$ there is an edge between $y$ and a vertex in $S\setminus X$. Clearly, $X\cup Y$ is part of the {\sc Vertex Cover} solution and it will cover all the edges incident on $S$. Then we are left with finding an optimum vertex cover in $G-(S\cup Y)$ which is  a chordal graph. This can be done in polynomial time. As we have $2^k$ choices for $X$, the total running time  of the algorithm is $2^{k}n^{\OO(1)}$. 
An FPT algorithm for {\sc FVS by CVD} is given by Jansen et al~\cite{jansen2014parameter} where they first find the modulator. This algorithm follows the algorithm to find a minimum FVS in bounded treewidth graphs and a similar trick works for {\sc Odd Cycle Transversal} too, when the modulator is given.  
% similar to the one for graphs of bounded treewidth. Using the similar trick, we can obtain an FPT algorithm for {\sc Odd Cycle Transversal by CVD} too.
However, the best known algorithm to find a CVD (modulator $S$) 
of size at most $k$ runs in time $k^{\OO(k)}n^{\OO(1)}$ \cite{cao2016chordal}. 
%So for \VCCVD\, if we first find a modulator and then apply the algorithm $\AAA$ where a modulator is part of the input, then the overall running time shoots up to $k^{\OO(k)}n^{\OO(1)}$. 

When the modulator is given, the FPT algorithms discussed above have been generalized for other problems and other classes of graphs (besides those that are $k$ away from the class of chordal graphs).
Let $\Phi$ be a Counting Monadic Second Order Logic (CMSO) formula and $t \geq 0$ be an integer. For
a given graph $G = (V, E)$, the task is to maximize $|X|$ subject to the following constraints: there is a
set $F \subseteq V$ such that $X \subseteq F$, the subgraph $G[F]$ induced by $F$ is of treewidth at most $t$,
and structure $(G[F], X)$ models $\Phi$. Note that the problem corresponds to finding minimum vertex cover and minimum feedback vertex set when $t=0$ and $t=1$ respectively when $\Phi$ is a tautology. For a polynomial $poly$, let $G_{poly}$ be the class of graphs such that, for any $G \in G_{poly}$, graph $G$ has at most $poly(n)$ minimal separators. 
% Let $G_{poly}$ be the class of graphs having at most polynomial in $n$ many minimal separators for some polynomial $poly$. 
 Fomin et al~\cite{DBLP:journals/siamcomp/FominTV15} gave a polynomial time algorithm for solving this optimization problem on the graph class $G_{poly}$. %in time $\mathcal{O}(|\Sigma_G| n^{t+4} f(t,\Phi))$ where $\Sigma_G$ is the set of all potential maximal cliques in $G$ and $f$ is a function of $t$ and $\Phi$ only. A potential maximal clique of $G$ is a vertex subset inducing a clique on some minimal triangulation of $G$ where a minimal triangulation of $G$ is a minimal chordal supergraph of $G$ (with the same vertex set).  
%Let $G_{poly}$ be the class of graphs having at most polynomial in $n$ many minimal separators for some polynomial $poly$. 
%It has been proved that the graphs in $G_{poly}$ admit polynomial number of potential maximal cliques~\cite{DBLP:journals/tcs/BouchitteT02}. 
Consider $G_{poly}+kv$ to be the graph class formed from $G_{poly}$ where to each graph we add at most $k$ vertices of arbitrary adjacencies. Liedloff et al.~\cite{DBLP:journals/algorithmica/LiedloffMT19} further proved that, the above %generic optimization 
problem is FPT on $G_{poly}+kv$, with parameter $k$, where the  modulator is also a part of input. As a chordal graph has polynomially many minimal separators~\cite{golumbic2004algorithmic}, we obtain that  \VCCVD\ and {\sc Feedback Vertex Set by CVD} are FPT when the modulator is given.

A possible way to solve these problems when modulator is not given is to obtain an approximation for the modulator (in this case CVD).
This is the approach that works for problems parameterized by treewidth.
For example, consider the {\sc Independent Set} problem parameterized by treewidth of the graph $\tw$. Using standard dynamic programming (DP), we can find a maximum independent set when we are given a tree decomposition of width $k$ as input in $2^{k} \cdot k^{O(1)} \cdot n$ time\cite{cygan2015parameterized}. But
the best known algorithm for outputting a tree-decomposition of minimum width takes time $\tw^{\OO(\tw^3)}n$ where $\tw$ is the treewidth of the given $n$-vertex graph\cite{bodlaender1996linear}. %On the other hand, using standard dynamic programming (DP), we can find a maximum independent set, when we are given a tree decomposition of width $k$ as input, in $2^{k} \cdot k^{O(1)} \cdot n$ \cite{cygan2015parameterized}. 
Thus, the total running time is $\tw^{\OO(\tw^3)}n$, when a tree decomposition is not given as an input. But we can overcome this issue by obtaining a tree decomposition of width $5\tw$ in time $2^{\OO(\tw)}n$~\cite{bodlaender2016c} and then applying the DP algorithm over the tree decomposition. 

We do not know of a constant factor (FPT) approximation algorithm for CVD even with $2^{\OO(k)}n^{\OO(1)}$ running time like in the case of treewidth. There are many recent results on polynomial time approximation algorithms for {\sc Chordal Vertex Deletion} \cite{jansen2017approximation,DBLP:conf/approx/AgrawalLM0Z18,kim2018erdHos} with the current best algorithm having a $\OO({\sf opt}\log {\sf opt})$ ratio, where ${\sf opt}$ is the size of minimum CVD~\cite{kim2018erdHos}. Thus, if we use this algorithm  along with algorithm $\AAA$, then the running time will be $2^{\OO(k^2\log k)}n^{\OO(1)}$. 

%In this paper, we initiate and study the version of these problems where the modulator is not given with the input. 
%Hence the algorithms we develop solve the problem in FPT time without really finding the modulator, and their running times are better than the running time  of algorithms that actually find the modulator. Basically the algorithms given either indicate that the size of modulator is larger than the size given in input or solve the problem. Thus, our algorithms are robust in the sense that they give correct output even when the input doesn't fall in desired restricted domain. Hence we don't need promise on the size of modulator for our algorithms. %Ideally, it would be nice to solve these problems even without the promise of the modulator size and expect the algorithms to detect in the same time if the modulator size is large. But we are not that far in our understanding of these problems.

One previous example we know of a parameterized problem where the FPT algorithm solves the problem without the modulator or even the promise, is {\sc Vertex Cover} parameterized by the size of {\sc K\"{o}nig Vertex Deletion} set $k$. A K\"{o}nig vertex deletion set of $G$ is a subset of vertices of $G$ whose removal results in a graph where the size of its minimum vertex cover and maximum matching are the same. In {\sc Vertex Cover by K\"{o}nig Vertex Deletion}, we are given graph $G=(V,E)$, $k, \ell \in \mathbb{N}$ and an assumption that there exists a K\"{o}nig vertex deletion set of size $k$ in $G$, here $k$ is parameter. We want to ask whether there exist a vertex cover of size $\ell$ in $G$? Lokshtanov et al. \cite{lokshtanov2014faster} solve {\sc Vertex Cover by K\"{o}nig Vertex Deletion} in $\OO^*(1.5214^k)$ time \footnote{$\mathcal{O}^{*}$ notation hides polynomial factor in the input length} without the promise. 

Finally we remark that there is an analogous line of work in the classical world of polynomial time algorithms. For example, it is known that finding a maximum clique in a unit disk graph is polynomial time solvable given a unit disk representation of the unit disk graph~\cite{cliqueUnitDisk}, though it is $NP$-hard to recognize whether a given graph is a unit disk graph~\cite{unitdiskhardrecognise}. Raghavan and Spinrad~\cite{robust} give a robust algorithm that given a graph either finds a maximum clique in the graph or outputs a certificate that the given graph is not a unit disk graph. See also~\cite{brandstadt2001robust, habib2005simple, DBLP:journals/siamcomp/FominTV15} for some other examples of robust algorithms. 

\smallskip
\noindent
\textbf{Our Techniques:}
The first step in our algorithms is to obtain, what we call a semi-clique tree decomposition of the given graph if one exists.
It is known~\cite{golumbic2004algorithmic} that every chordal graph has a clique-tree decomposition, i.e., a tree decomposition where every bag is a clique in the graph. If the modulator is given, then we can add it to each bag, and obtain a tree-decomposition where each bag is a clique plus at most $k$ vertices. In our case (where the modulator is not given), we obtain 
a tree decomposition in $2^{\OO(k)}n^{\OO(1)}$ where each bag can be partitioned into $C\uplus N$, where $C$ can be covered by at most $4$ cliques in $G$ and $\vert N\vert \leq 7k+5$. 
Here we also know a partition $C_1\uplus C_2\uplus C_3 \uplus C_4$ of $C$ where each $C_i$ is a clique.  We call this tree decomposition a $(4,7k+5)$-semi clique tree decomposition. 
%That is a $(4,7k+5)$-semi clique tree decomposition of a graph $G$ is a tree decomposition such that each bag can be partitioned into $C_1\uplus C_2\uplus C_3\uplus C_4 \uplus N$ where each $C_i$ is a clique in $G$ and $\vert N\vert \leq 7k+5$. 
 Our result in this regard is formalized %as follows. 
in the following theorem.

\begin{theorem} \label{theorem:treedecomp}
There is an algorithm that given a graph $G$ and an integer $k$ runs in time $\OO(2^{7k} \cdot (kn^{4} +n^{\omega+2}))$ where $\omega$ is the matrix multiplication exponent and either constructs a $(4,7k+5)$-semi clique tree decomposition $\mathcal{T}$ of $G$ or concludes that there is no chordal vertex deletion set of size $k$ in $G$. Moreover, the algorithm also provides a partition $C_1\uplus C_2\uplus C_3\uplus C_4 \uplus N$ of each bag of ${\cal T}$ such that  $\vert N\vert \leq 7k+5$ and $C_i$ is a clique in $G$ for all $i\in \{1,2,3,4\}$. 
\end{theorem}
%Here $\omega$ is the exponent of the best matrix multiplication algorithm.

After getting a $(4,7k+5)$-semi clique tree decomposition, we then design DP algorithms for {\sc Vertex Cover}, {\sc Feedback Vertex Set} and {\sc Odd Cycle Transversal} on this tree decomposition. Since the vertex cover of a clique has to contain all but one vertex of the clique, the number of ways the solution might intersect a bag of the tree is at most $\OO(2^{7k}n^4)$. Using this fact, one can bound the running time for the DP algorithm to $\OO(2^{7k}n^5)$. The overall running time would be the sum of the time taken to construct a $(4,7k+5)$-semi clique tree decomposition  and the time of the DP algorithm on this tree decomposition which is bounded by $\OO(2^{7k}n^{5})$. In the case of {\sc Feedback Vertex Set} and {\sc Odd Cycle Transversal}, again from each clique all but two vertices will be in the solution. Using this fact one can bound the running time of {\sc FVS By CVD} and {\sc OCT by CVD} to be $\OO^*(2^{O(k)})$. 

We like to add that the algorithms obtained are robust due to Theorem~\ref{theorem:treedecomp}.
%We also show that all the problems mentioned above cannot be solved in $\OO^*((2- \epsilon)^{k})$ time  under Strong Exponential Time Hypothesis (SETH) even if a CVD of size $k$ is given as part of the input. 
%This matches the upper bound of algorithm $\AAA$ for {\sc Vertex Cover}.

\smallskip
\noindent
\textbf{Organization of the paper:} In Section \ref{sec:Prelims}, we state  graph theoretic notations used in this paper and give the necessary preliminaries on tree decomposition and parameterized complexity. In Section \ref{sec:Tree Decom}, 
we prove Theorem~\ref{theorem:treedecomp}. 
%we present the constructing special tree decomposition of the graph such that every node in the decomposition is a union of at most $3$ cliques plus $O(k)$ vertices. 
In Section \ref{sec:Structural Param}, we give algorithms for problems \VCCVD, {\sc FVS by CVD} and {\sc OCT by CVD} using dynamic programming on semi clique tree decomposition and lower bounds for these problems assuming SETH. %to appendix (Section \ref{appendix:lowerbounds}) due to paucity of space.
%\todo[inline]{modify after setting up all sections}

%!TEX root = main_lipics.tex

\section{Preliminaries} \label{sec:Prelims}

For $n\in {\mathbb N}$, $[n]$ denotes the set $\{1,\ldots,n\}$. 
We use $A \uplus B$ to denote the set formed from the union of disjoint sets $A$ and $B$. For a function
$w : X \rightarrow \mathbb{R}$, we use $w(D) = \sum_{x \in D}^{} w(x)$.

%Given a graph $G =(V, E)$, $V$ and $E$ denote its vertex-set and edge-set, respectively. 
We use the term graph for a simple undirected graph without loops and parallel edges. 
For a graph $G$, we use $V(G)$ and $E(G)$ to denote its vertex set and edge set, respectively. 
Let  $G =(V, E)$ be a graph.  
For $V'\subseteq V$,  $G[V']$ and $G-V'$ denote the graph induced on $V'$ and $V\setminus V'$, respectively.  
%For sets $V$ and $V$, we denote $V \setminus V'$ as $V - V'$ .  
For a vertex $v \in V$, $G - v$ denotes the graph $G-\{v\}$.
%denotes graph induced on vertex set $V - v$. 
 For a vertex $v \in V$, $N_G(v)$ and $N_G[v]$ denote the open neighborhood and closed neighborhood  of $v$, respectively. That is, $N_G(v)=\{u \colon \{v,u\} \in E\}$ and $N_G[v]=N_G(v) \cup \{v\}$. Also we define for a subset $X \subseteq V(G)$, $N_G(X) = \bigcup_{v \in X} (N_G(v) \setminus X)$ and $N_G[X] = N_G(X) \cup X$. We omit the subscript $G$, when the graph is clear from the context.   
 A graph is \textit{chordal} if it does not contain a cycle of length greater than or equal to $4$ as an induced subgraph. 
 %A chordless cycle of length greater than $3$ is called a \textit{hole}. 
 A subset $S \subseteq V(G)$ such that $G - S$ is a chordal graph is called the \textit{chordal vertex deletion set}.
We say that a graph $G$ is a union of $\ell$ cliques if $V(G) = V_1\uplus \ldots \uplus V_{\ell}$ and $V_i$ is a clique in $G$ for all $i\in \{1,\ldots,\ell\}$. 
We use standard notation and terminology from the book~\cite{Diestel} for graph-related terms which are not explicitly defined here.

Next we define separator, separation and tree decomposition in graphs and finally we define our new notion of special tree decomposition  which we call $(c,\ell)$-semi clique tree decomposition where $c,\ell \in {\mathbb N}$.   

\begin{definition}[Separator]
Given a graph $G$ and vertex subsets $A,B \subseteq V(G)$, a subset $C \subseteq V(G)$ is called a separator of $A$ and $B$ if every path from a vertex in $A$ to a vertex in $B$ (we call it $A-B$ path) contains a vertex from $C$.
\end{definition}

\begin{definition}[Separation]
For a graph $G$, a pair of vertex subsets $(A,B)$ is a separation in $G$ if $A \cup B = V(G)$ and $A \cap B$ is a separator of $A \setminus B$ and $B \setminus A$.
\end{definition}

\begin{definition}[Balanced Separator]
For a graph $G$, a weight function $w: V(G) \rightarrow \mathbb{R}_{\geq 0}$ and $0<\alpha<1$, a set $S \subseteq V(G)$ is called an $\alpha$-balanced separator of $G$ with respect to $w$ if for any connected component $C$ of $ G - S$,  $ w(V(C)) \leq \alpha \cdot w(V(G))$. %\sum\limits_{v \in V(C)}^{} w(v) \leq \frac{2}{3}\sum\limits_{v \in V}^{} w(v) $.
\end{definition}

\begin{definition}[Balanced Separation]
Given a graph $G$, a weight function $w: V(G) \rightarrow \mathbb{R}_{\geq 0}$, and $0<\alpha<1$, a pair of vertex subsets $(A,B)$ is an $\alpha$-balanced separation in $G$ with respect to $w$ if  $(A,B)$ is a separation in $G$ and 
 $ w(A \setminus B) \leq \alpha \cdot w(V(G))$ and  $ w(B \setminus A) \leq \alpha \cdot w(V(G))$. %$\sum\limits_{v \in A}^{} w(v) \leq \frac{2}{3}\sum\limits_{v \in V}^{} w(v) $ and $\sum\limits_{v \in B}^{} w(v) \leq \frac{2}{3}\sum\limits_{v \in V}^{} w(v) $.
\end{definition}

%\begin{definition}[Tree decomposition]
%Given a graph $G = (V, E)$, a tree decomposition is a pair $(X, T)$, where $X = \{X_1, \dotsc , X_n\}$ is a family of subsets of $V$, and $T$ is a tree whose nodes are the subsets $X_i$, satisfying the following properties:
%\begin{itemize}
%\item $\bigcup\limits_{i=1}^{n} X_i = V$.
%\item For all edges $(u,v) \in E$, there is a subset $X_i$ that contains both $u$ and $v$.
%\item If $X_i$ and $X_j$ both contain a vertex $v$, then all nodes $X_k$ of the tree in the (unique) path between $X_i$ and $X_j$ contain $v$ as well.
%\end{itemize}
%The sets $X_i$ are called the bags corresponding to node $i$.
%\end{definition}
%
%The treewidth of a graph is $ \min\limits_{(X,T)} \max\limits_{X_t} |X_t|$ over all possible tree decompositions $(X,T)$.

\begin{definition}[Tree decomposition]
A tree decomposition of a graph $G$ is a pair ${\cal T}=(T,\{X_t\}_{t\in V(T)})$, 
where $T$ is a tree and for any $t\in V(T)$, a vertex subset $X_t\subseteq V(G)$ is associated with it, called a {\em bag}, such that the following conditions holds. 
\begin{itemize}
\item $\bigcup_{t\in V(T)} X_t = V(G)$.
\item For any edge $\{u,v\} \in E(G)$, there is a node $t\in V(T)$ such that $u,v\in X_t$.
\item For any vertex $u\in V(G)$, the set $\{t \in V(T) \colon u\in X_t\}$ of nodes induces a connected subtree of $T$. 
%If $X_i$ and $X_j$ both contain a vertex $v$, then all nodes $X_k$ of the tree in the (unique) path between $X_i$ and $X_j$ contain $v$ as well.
\end{itemize}
The width of the tree decomposition ${\cal T}$ is $\max_{t\in V(T)} \vert X_t\vert -1$ and the treewidth of $G$ is the minimum width over all tree decompositions of $G$. 
\end{definition}

\begin{proposition}[\cite{downey2013fundamentals}]
\label{prop:clique_in_a_bag}
%\todo{cite}
Let $G$ be a graph and $C$ be a clique in $G$. Let ${\cal T}=(T,\{X_t\}_{t\in V(T)})$ be a tree decomposition of $G$. Then, there is a node $t\in V(T)$  such that $C\subseteq X_t$. 
\end{proposition}

%Recall that in a tree decomposition  ${\cal T}=(T,\{X_t\}_{t\in V(T)})$ of $G$, every bag $X_t$ is a separator.   

%
%
%\begin{definition}[Nice Tree decomposition]
%A nice tree decomposition is a tree decomposition satisfying the following properties:
%\begin{itemize}
%\item Let us arbitrarily root the tree $T$. For the root of the tree $r$, $ X_r= \phi $.
%\item $X_l = \phi$ for all the leaf nodes of the tree.
%\item Every other node of $T$ are of three types:
%\begin{itemize}
%\item Introduce Node: A node $i$ with exactly one child $j$ such that $X_i = X_j \cup \{v\}$ for some vertex $v \notin X_j$.
%\item Forget Node: A node $i$ with exactly one child $j$ such that $X_i = X_j \setminus \{v\}$ for some vertex $v \in X_j$.
%\item Join Node: A node $i$ with exactly two children $j$ and $j'$ such that $X_i = X_j =X_{j'}$.
%
%\end{itemize}
%\end{itemize}
%\end{definition} 

%
\begin{definition}[Clique tree decomposition]
A clique tree decomposition of a graph $G$ is a tree decomposition ${\cal T}=(T,\{X_t\}_{t\in V(T)})$
where $X_t$ is a clique in $G$ for all $t\in V(T)$.  
%corresponding to node $t$ in the decomposition, $G[B_t]$ is a clique.
\end{definition}

%%%%Vibha
\begin{proposition} [\cite{golumbic2004algorithmic}]A graph is chordal if and only if it has a clique tree decomposition.
\end{proposition}

%For the graphs $G$ having a chordal vertex deletion set $X$ of size at most $k$, we can construct a tree decomposition by adding $X$ to each bag of the clique-tree decomposition of $G - X$. The following lemma can be proved using standard arguments on trees on this tree decomposition%~\cite{cygan2015parameterized}.

Next we define the notion of $(c,\ell)$-semi clique and then define $(c,\ell)$-semi clique tree decomposition.

\begin{definition}
%[$(c,\ell)$-semi clique]
%A subset $S \subseteq V(G)$ is 
A graph $G$ is called an $(c,\ell)$-semi clique if there is a partition $C\uplus N$ of $V(G)$ such that $G[C]$ is a union of at most $c$ cliques and $|N| \leq \ell$.
\end{definition}

%We define a new type of tree decomposition called $(\alpha,\beta)$ semi clique-tree decomposition.

\begin{definition}[$(c,\ell)$-semi clique tree decomposition]
For a graph $G$ and $c,\ell\in {\mathbb N}$,  
a tree decomposition ${\cal T}=(T,\{X_t\}_{t\in V(T)})$ of $G$ is a 
$(c,\ell)$-semi clique tree decomposition if  $G[X_t]$ is a $(c,\ell)$-semi clique for each $t\in V(T)$. 
% of $G[B_t]$.% = C_t \cup N_t$ corresponding to node $t$ in the decomposition, $G[C_t]$ is the union of at most $\alpha$ cliques and $|N| \leq \beta$.
\end{definition}

We use the following lemma in Section~\ref{sec:Tree Decom}.

\begin{proposition}[\cite{fomin2015parameterized}] \label{lemma:tree-sep}
Let $T$ be a tree and  $x,y,z \in V(T)$. Then there exists a vertex $v \in V(T)$ such that every connected component of $T - v $ has at most one vertex from $\{x, y,z\}$.
\end{proposition}

%The following lemma can be easily shown using the same ideas used to convert a tree decomposition to a nice tree decomposition \cite{cygan2015parameterized}.

%\todo[inline]{The definition of nice tree decomposition and this lemma may go to the application section that will be in appendix, except an overview}
%
%\begin{lemma} \label{lemma:nice-tree}
%There is an algorithm that given input a $(\alpha,\beta)$ semi-clique tree decomposition, %where each node X can be partitioned into clique part $C$ and non-clique part N where $N \leq 6k+4$, 
% constructs a nice $(\alpha,\beta)$ semi-clique tree decomposition with with $O(n)$ nodes.
%\end{lemma}

%We use the following conjecture to prove lower bounds.
%\begin{conjecture}[Strong Exponential Time Hypothesis (SETH)(~\cite{IPZ01})]
%\label{thm:eth}
%{\sc CNF-SAT} cannot be solved in $\OO^*((2- \epsilon)^{n})$ time where the input formula has $n$ variables and $m$ clauses.
%\end{conjecture}

For definitions and notions on parameterized complexity, we refer to \cite{cygan2015parameterized}.

%We define the following problem for later use.

%\defparproblem{{\sc Vertex Cover by ClsVD}}{A graph $G=(V,E)$, $k, \ell \in \mathbb{N}$ and a set $S \subseteq V(G)$ with $|S|\leq k$ such that $G[V \setminus S]$ is a cluster graph.}{$k$}{
%Is there a vertex cover of size $\ell$ in $G$?
%%Does there exist a set $S \subseteq V(G)$ of at most $\ell$ vertices such that $G[V-S]$ is edgeless graph?
%}

 \medskip\noindent\textbf{SETH.}
 For $q\geq 3$, let $\delta_q$ be the infimum of the set of constants $c$ for which there exists an algorithm solving $q$-SAT with $n$ variables and $m$ clauses in time $2^{cn}\cdot m^{\OO(1)}$.  
%The {\em{Exponential-Time Hypothesis} (ETH)} and {\em{Strong Exponential-Time Hypothesis} (SETH)} are then formally defined as follows.
%ETH conjectures  that $\delta_3>0$ and  
The {\em{Strong Exponential-Time Hypothesis} (SETH)}
conjectures  that $\lim_{q\to \infty}\delta_q=1$. SETH implies that {\sc CNF-SAT} on $n$ variables cannot be solved in $\OO^*((2- \epsilon)^{n})$ time for any $\epsilon>0$.

We define {\sc Node Multiway Cut} problem where we are given an input graph $G = (V, E)$, a set $T \subseteq V$ of terminals and an integer $k$. We want to ask whether there exist a set $X \subseteq V \setminus T$ of size at most $k$ such that any path between two different terminals intersects $X$.

\section{Semi Clique Tree Decomposition} \label{sec:Tree Decom}
Given a graph $G$ such that it contains a CVD of size $k$, our aim is to construct a $(4, 7k+5)$-semi clique tree decomposition $\mathcal{T}$ of $G$. %We later apply dynamic programming for {\sc Vertex Cover} and {\sc Feedback Vertex Set} on this tree decomposition $\mathcal{T}$ to get algorithms with $2^{\OO(k)}n^{\OO(1)}$ running time.
 We loosely follow the ideas used for the tree decomposition algorithm in ~\cite{cygan2015parameterized} to construct a tree decomposition of a graph $G$ of width at most $4\tw(G)+4$, where $\tw(G)$ is the tree-width of $G$. 
But before that we propose the following lemmas that we use in getting the required $(4, 7k+5)$-semi clique tree decomposition.

\begin{lemma}
\label{lem:existence}
Let $G$ be a graph having a CVD of size $k$. Then $G$ has a $(1,k)$-semi clique tree decomposition. 
\end{lemma}

%\begin{proof}
%Let $Y$ be the chordal vertex deletion set of $G$ of size  $k$. Since $G - Y$ is a chordal graph, it has a clique tree decomposition ${\cal T}'$. Adding $Y$ to each bag of the tree decomposition ${\cal T}'$, we get a $(1,k)$-semi clique tree decomposition ${\cal T}=(T,\{X_t\}_{t\in V(T)})$ of $G$. 
%\end{proof}

\begin{lemma} \label{lemma:maxicliq}
For a graph $G$ on $n$ vertices with a CVD of size $k$, the number of maximal cliques in $G$ are bounded by $\OO(2^k \cdot n)$. Furthermore, there is an algorithm that given any graph $G$ either concludes that there is no CVD of size $k$ in $G$ or enumerates all the the maximal cliques of $G$ in $\OO(2^k \cdot n^{\omega+1})$ time where $\omega$ is the matrix multiplication exponent. % can be enumerated with $\OO(n^{\omega})$  delay (the maximum time taken between outputting two consecutive solutions).
\end{lemma}
\begin{proof}
%Let $n=\vert V(G)\vert$. 
%We know that  there exists a vertex subset 
Let $X\subseteq V(G)$ be of size at most $k$ such that $G - X$ is a chordal graph. For any maximal clique $C$ in $G$ 
%such that $C = C_{X} \cup C_{G - X}$ where 
let $C_{X}=C\cap X$ and $C_{G - X}=C\setminus X$.
% are the vertex sets contained in $X$ and $V \setminus X$ respectively. 
Since $ G - X$ is a chordal graph, it has only $\OO(n)$ maximal cliques ~\cite{golumbic2004algorithmic}. 
%Hence for each fixed subset $C_X\subseteq X$, there are $\OO(n)$ sets $C'_{G - X}$ such that $C'_{G - X}$ is a maximal clique in $G - X$. 

We claim that for a subset $C_X \subseteq X$ and a maximal clique $Q$ in $G-X$, there is at most one subset $Q'\subseteq Q$ such that  $C_X \cup Q'$ forms a maximal clique in $G$. 
If there are two distinct subsets $Q_1, Q_2$ of $Q$ such that $C_X\cup Q_1$ and $C_X\cup Q_2$ are cliques in $G$, then $C_X\cup Q_1\cup Q_2$ is a clique larger than the cliques $C_X\cup Q_1$ and $C_X\cup Q_2$. 
Thus, since there are at most $2^k$ subsets of $X$ and at most $\OO(n)$ maximal cliques in $G$, the total number of maximal cliques in $G$ is upper bounded by $\OO(2^kn)$. 

%that does so with $D \neq D'$. Look at element $x \in D \setminus D'$. Since $G[C_X \cup D]$ forms a clique, $x$ is adjacent to all the vertices of $C_X$. Since $G[C'_{G - X}]$ forms a clique, $x$ is adjacent to all the vertices in $D'$. Then $G[ C_X \cup D' \cup x]$ forms a clique as well contradicting the fact that $G[ C_X \cup D']$ forms a maximal clique in $G$.

%there can be at most one subset $C_{G - X} \subseteq C'_{G - X}$ such that 
%$C_X \cup C_{G - X}$ forms a maximal clique in $G$. Suppose there are two subsets $D, D'$ of $C'_{G - X}$ that does so with $D \neq D'$. Look at element $x \in D \setminus D'$. Since $G[C_X \cup D]$ forms a clique, $x$ is adjacent to all the vertices of $C_X$. Since $G[C'_{G - X}]$ forms a clique, $x$ is adjacent to all the vertices in $D'$. Then $G[ C_X \cup D' \cup x]$ forms a clique as well contradicting the fact that $G[ C_X \cup D']$ forms a maximal clique in $G$.

%Since there are at most $2^k$ sets $C_X$ which are subsets of $X$ and each $C_X$ has  $O(n-k)$ sets $C_{G - X}$ such that $C_X \cup C_{G - X}$ forms a maximal clique in $G$, there are at most $\OO(2^k \cdot n)$ maximal cliques in $G$.

There is an algorithm that given a graph $H$, enumerates all the maximal cliques of $H$ with $\OO(\vert V(H)\vert^{\omega})$ delay(the maximum time taken between outputting two consecutive solutions)~\cite{makino2004new}. If $G$ has a CVD of size $k$, there are at most $\OO(2^k \cdot n)$ maximal cliques in $G$ which can be enumerated in $\OO(2^k \cdot n^{\omega+1})$ time. Else, we note that the number of maximal cliques enumerated is more than $\OO(2^k \cdot n)$ and hence return that $G$ has no CVD of size $k$.
\end{proof}

\begin{lemma} \label{lemma:balsep}
Let $G$ be a graph having a CVD of size $k$ and $w: V(G) \rightarrow \mathbb{R}_{\geq 0}$ be a weight function on $V(G)$. There exists a $\frac{2}{3}$-balanced separation $(A,B)$ of $G$ with respect to $w$ such that the graph induced on the corresponding separator $G[A \cap B]$ is a $(1,k)$-semi clique.
% = C \uplus N$ where $G[C]$ is a clique and $|N| \leq k$.
\end{lemma}
\begin{proof}
First we  prove that there is a $\frac{1}{2}$-balanced separator $X$ such that $G[X]$ is a $(1,k)$-semi clique. 
%Let $Y$ be the chordal vertex deletion set of $G$ of size  $k$. Since $G - Y$ is a chordal graph, it has a clique-tree decomposition ${\cal T}'$. Adding $Y$ to each bag of the tree decomposition ${\cal T}'$, we get 
By Lemma~\ref{lem:existence}, there is a $(1,k)$-semi clique tree decomposition ${\cal T}=(T,\{X_t\}_{t\in V(T)})$ of $G$. 
Arbitrarily root the tree of $T$ at a node $r \in V(T)$. For any node $y \in V(T)$,  let $T_y$ denote the subtree of $T$ rooted at node $y$ and $G_y$ denote the graph induced on the vertices of $G$ present in the bags of nodes of $T_y$. That is $V(G_y)=\bigcup_{t\in V(T_y)}X_t$. 
Let $t$ be the farthest node of $T$ from the root $r$
such that $w(V(G_t)) > \frac{1}{2} w(V(G))$. 
That is, for all nodes $t' \in V(T_t)\setminus \{t\}$, we have that  
$w(V(G_{t'})) \leq \frac{1}{2} w(V(G))$.
%$\sum\limits_{v \in V(G_x)}^{} w(v) \leq \frac{1}{2}\sum\limits_{v \in V}^{} w(v) $.

We claim that $X=X_t$ is a  $\frac{1}{2}$-balanced separator of $G$. 
Let $t_1, \dotsc, t_p$ be the children of $t$. Since $X$ is a bag of the tree decomposition ${\cal T}$, all the connected components of $G - X$ are contained either in $G_{t_i} - X$ or $G[V(G)\setminus V(G_t)]$. Since 
$w(V(G_t)) > \frac{1}{2} w(V(G))$,
%$\sum\limits_{v \in V(G_x)}^{} w(v) > \frac{1}{2}\sum\limits_{v \in V}^{} w(v) $,
 we have $w(V(G) \setminus V(G_x)) < \frac{1}{2} w(V(G))$.
% $\sum\limits_{v \in V \setminus V(G_x)}^{} w(v) \leq \frac{1}{2}\sum\limits_{v \in V}^{} w(v) $. 
 By the choice of $t$, we have $w(V(G_{t_i})) \leq \frac{1}{2} w(V(G))$ for all $i \in [p]$.
% $\sum\limits_{v \in V(G_{x_i})}^{} w(v) \leq \frac{1}{2}\sum\limits_{v \in V}^{} w(v) $.

Now we define a $\frac{2}{3}$-balanced separation $(A,B)$ for $G$ such that the set $X=A\cap B$ ($\frac{1}{2}$ balanced separator).
% is the corresponding separator.
%using this balanced separator $X$, we come up with a balanced separation $X'$ for $G$. 
Let $D_1, \dotsc, D_q$ be the vertex sets of the connected components of $G - X$.
Let $a_i = w(D_i)$ for all $i\in [q]$. Without loss of generality, assume  that $a_1 \geq \dotsc \geq a_q$. Let $q'$ be the smallest index such that $\sum_{i=1}^{q'} a_i \geq \frac{1}{3} w(V(G))$ or $q'=q$ if no such index exists. 
Clearly, $\sum_{i=q'+1}^{q} a_i \leq \frac{2}{3} w(V(G))$. We prove that $\sum_{i=1}^{q'} a_i \leq \frac{2}{3} w(V(G))$.
If $q'=1$, $\sum_{i=1}^{q'} a_i = a_{q'} \leq \frac{1}{2} w(V(G))$ and we are done. Else, since $q'$ is the smallest index such that $\sum_{i=1}^{q'} a_i \geq \frac{1}{3} w(V(G))$, we have $\sum_{i=1}^{q'-1} a_i < \frac{1}{3} w(V(G))$. We also note that $a_{q'} \leq a_{q'-1} \leq \sum_{i=1}^{q'-1} a_i < \frac{1}{3} w(V(G))$. Hence  $\sum_{i=1}^{q'} a_i = \sum_{i=1}^{q'-1} a_i + a_{q'} \leq \frac{2}{3} w(V(G))$.% and $q$ is the smallest index with $\sum_{i=1}^{q} a_i \geq \frac{1}{3} w(V(G))$, we have that $\sum_{i=1}^{q} a_i \leq \frac{2}{3} w(V(G))$. 

%When $q=r$, we have $\sum\limits_{i=1}^{r} a_i \leq \frac{1}{3} w(V(G))$ and the claim holds. 
% 
% When $q=1$, we have $\sum\limits_{i=1}^{q} a_i = a_1 \leq \frac{1}{2} w(V)$ and the claim holds. Hence we can assume that $q>1$ and $\sum\limits_{i=1}^{q} a_i \geq \frac{1}{3} w(V)$. Since $q$ is the smallest index such that $\sum\limits_{i=1}^{q} a_i \geq \frac{1}{3} w(V)$, we have $\sum\limits_{i=1}^{q-1} a_i < \frac{1}{3} w(V)$. Also we have $a_q \leq a_{q-1} \leq \sum\limits_{i=1}^{q-1} a_i <\frac{1}{3} w(V)$. Hence we have $\sum\limits_{i=1}^{q} a_i = \sum\limits_{i=1}^{q-1} a_i + a_q \leq \frac{2}{3} w(V)$.
%
Now we define $A = X\cup \bigcup_{ i \in [q]} D_i$ and $B = X\cup \bigcup_{ i \in [q] \setminus [q']} D_i$. 
Notice that $X=A\cap B$ and $(A,B)$ is a separation of $G$. 
%
%We claim that $(A,B)$ is a balanced separation with $X = A \cap B$ being the corresponding separator. This is true as 
Also notice that 
$w(A \setminus B) = \sum_{i=1}^{q'} a_i \leq \frac{2}{3} w(V)$ and $w(B \setminus A) = \sum_{i=q'+1}^{q} a_i \leq w(V(G)) - \frac{1}{3} w(V(G)) = \frac{2}{3}w(V(G))$ as $\sum\limits_{i=1}^{q'} a_i \geq \frac{1}{3} w(V(G))$. 
Since $X$ is a bag of the tree decomposition ${\cal T}$, $G[X]$ is a $(1,k)$-semi clique. 
%it can be partitioned as $X = C \uplus N$ where $G[C]$ is a clique and $|N| \leq k$. 
\end{proof}

Using Lemmas \ref{lemma:maxicliq} and \ref{lemma:balsep}, we obtain the following corollary.

\begin{corollary} \label{corollary:separation4k}
Let $G$ be a graph with a CVD of size $k$. Let $N \subseteq V(G)$ %with $S = C \uplus N$ where $G[C]$ is a union of $d$ cliques with $d \in [3]$ and
 with $5k+3 \leq |N| \leq 6k+4$. Then there exists a partition $(N_A, N_B)$ of $N$ and a vertex subset $X \subseteq V(G)$ satisfying the following properties.
\begin{itemize}
\item $|N_A|, |N_B| \leq 4k+2$. 
\item $X$ is a vertex separator of $N_A$ and $N_B$ in the graph $G$.
\item $G[X]$ is a $(1,k)$-semi clique.
%Also given the 2-partition $(N_A, N_B)$, we can find a subset $X$ satisfying the above properties in $\OO^*(2^k)$ time. 
\end{itemize}
Moreover, there is an algorithm that given any graph $G$, either concludes that there is no CVD of size $k$ in $G$ or computes such a partition $(N_A, N_B)$ of $N$ and the set $X$ 
%$C\uplus N'=X$, where $\vert N'\vert \leq k$ and $C$ is a clique in $G$,
 in $\OO(2^{7k} \cdot (kn^{3} +n^{\omega+1}) )$ time.
\end{corollary}
\begin{proof}
Let us define a weight function $w: V(G) \rightarrow \mathbb{R}_{\ge 0}$ such that $w(v)=1$ if $v \in N$ and $0$ otherwise. From Lemma \ref{lemma:balsep}, we know that there exists a pair of vertex subsets $(A,B)$ which is the balanced separation of $G$ with respect to $w$ where the graph induced on the corresponding separator $G[A \cap B]$ is a $(1,k)$ semi clique. %Let $X = C'' \uplus N''$ where $G[C'']$ is a clique and $|N''| \leq k$. Also let $(A,B)$ be the corresponding 2-partition of $V \setminus X$.

%We use Lemma~\ref{lemma:maxicliq} to go over all maximal cliques of $G$ to find a maximal clique $D$ such that $C''$ (the clique part of $X$) is a subset of $D$. We know that in the remaining graph $G[V \setminus D]$, there exists a balanced separation $N'' = X \setminus C''$ of size at most $k$. 

Let us define the partition $(N_A, N_B)$. We add $(A \setminus B) \cap N$ to $N_A$ and $(B \setminus A) \cap N$ to $N_B$. Since $(A,B)$ is a balanced separation of $G$ with respect to $w$, $|(A \setminus B) \cap N|, |(B \setminus A) \cap N| \leq \frac{2}{3}|N| \leq 4k+2$. % For each vertex $u \in (A \cap B) \cap N$, we iteratively add $u$ to the currently smaller of the two sets of $N_A$ and $N_B$. Since $|N| \leq 6k+4 \leq 2 \cdot (4k+2)$, we have $|N_A|, |N_B| \leq 4k+2$ even after this process. 
%Also since $|N| \geq 5k+3$ and $|N_A|, |N_B| \leq 4k+2$, we have  $|N_A|, |N_B| \geq k+1$. 
This shows the existence of subsets $N_A, N_B$ and $X= A \cap B$. But the proof is not constructive as the existence of $(A,B)$ uses the $(1,k)$-semi clique tree decomposition of $G$ which requires the chordal vertex deletion.

%-------------------------

%\todo[inline]{explain its computation}

We now explain how to compute these subsets without the knowledge of a $(1,k)$-semi clique tree decomposition of $G$. Let $X = C'' \uplus N''$ where $C''$ is a clique and $|N''| \leq k$. %Also let $(A,B)$ be the corresponding 2-partition of $W \setminus X$.
We use Lemma~\ref{lemma:maxicliq} to either conclude that $G$ has no CVD of size $k$ or go over all maximal cliques of $G$ to find a maximal clique $D$ such that $C'' \subseteq D$. %(the clique part of $X$) is a subset of $D$.
 We can conclude that in the remaining graph $G[V \setminus D]$, there exists a separator $Z \subseteq N'' = X \setminus C''$ of size at most $k$ for the sets $N_A$ and $N_B$. 

%Define $S_A$ as $A \cap S$ and $S_B$ as $B \cap S$. Also define $N_A = S_A \cap N$  and $N_B = S_B \cap N$. Since $X$ is a balanced separation with respect to $w$, we have $|N_A|, |N_B| \leq \frac{2}{3}|N| \leq \frac{2}{3}(6k+4) \leq 4k+2$. %Also $|N_A|, |N_B| \geq |N| - \frac{2}{3}|N| = 2k+2$.

We go over all $2^{|N|} \leq 2^{6k+4}$ 2-partitions of $N$ to guess the partition $(N_A,N_B)$. Then we apply the classic Ford-Fulkerson maximum flow algorithm to find the separator $Z$ of the sets $N_A$ and $N_B$ in the graph $G[V \setminus D]$. If $|Z| > k$, we can conclude that $G$ has no CVD of size $k$ in $G$. %By the existence of $N''$, we know that $|Z| \leq k$.
Thus, we obtained a set $X' = D \uplus Z$ such that $G[X']$ is a $(1,k)$-semi clique and $X'$ is a vertex separator of $N_A$ and $N_B$ in the graph $G$.

Now we estimate the time taken to obtain these sets. We first go over all $\OO(2^k \cdot n)$ maximal cliques of the graph which takes $\OO(2^k \cdot n^{\omega+1})$ time. Then for each of the $\OO(2^k \cdot n)$ maximal cliques, we go over at most $2^{6k+4}$ guesses for $N_A$ and $N_B$. Finally we use the Ford-Fulkerson maximum flow algorithm to find the separator of size at most $k$ for $N_A$ and $N_B$ which takes $\OO(k(n+m))$ time. %\cite{cygan2015parameterized}.
 Overall the running time is $\OO(2^k \cdot n^{\omega+1} + (2^kn) \cdot 2^{6k} \cdot (k(n+m))) = \OO(2^{7k} \cdot (kn^{3} +n^{\omega+1}) )$.% $\OO(\big(2^k \cdot n^{\omega+1}\big) \cdot \big( 2^{6k} (k(n+m)) \big) ) = \OO(2^{7k} \cdot n^{\omega+4})$.
\end{proof}

\begin{lemma} \label{lemma:3cliqsep}
Let $G$ be a graph having a CVD of size $k$. Let $C_1, C_2, C_3$ be three distinct cliques in $G$. Then there exists a vertex subset $X\subseteq V(G)$ such that $G[X]$ is a $(1,k)$-semi clique
and  $X$ is a separator of  $C_i$ and $ C_j$ for all $i,j\in \{1,2,3\}$ and $i\neq j$. 
Moreover, there is an algorithm that given any graph $G$, either concludes that there is no CVD of size $k$ in $G$ or computes $X$ in $\OO(4^{k} \cdot (kn^{3} +n^{\omega+1}) )$ time.
\end{lemma}

\begin{proof}

By Lemma~\ref{lem:existence}, there is a $(1,k)$-semi clique tree decomposition ${\cal T}=(T,\{X_t\}_{t\in V(T)})$ of $G$. By Proposition~\ref{prop:clique_in_a_bag}, we know that 
there exist nodes $t_1,t_2,t_3\in V(T)$ such that $C_1  \subseteq X_{t_1}$,  $C_2  \subseteq X_{t_2}$ and $C_3 \subseteq X_{t_3}$. From Proposition~\ref{lemma:tree-sep}, we know that there exists a node $t\in V(T)$ such that 
 $(i)$ $t_1$, $t_2$ and $t_3$ are in different connected components of $T-t$. We claim that $X=X_t$ is the required separator. Since $X$ is a bag in the $(1,k)$-semi clique tree decomposition ${\cal T}$, $G[X]$ is a $(1,k)$-semi clique. Because of statement $(i)$, we have that $X$ is a separator of  $C_i$ and $ C_j$ for all $i,j\in \{1,2,3\}$ and $i\neq j$. The proof is not constructive as we do not have a $(1,k)$-semi clique tree decomposition of $G$.
% 
%
%Let $Y$ be a chordal vertex deletion set of $G$ of size at most $k$. The chordal graph $G-Y$ has a clique-tree decomposition $\mathcal{T}$. We could also assume that $\mathcal{T}$ is such that every bag in $\mathcal{T}$ corresponds to vertex sets of maximal cliques of the graph $G[V \setminus Y]$ \cite{golumbic2004algorithmic}. Now the sets $C_1 \setminus Y$,  $C_2 \setminus Y$ and $C_3 \setminus Y$ are vertex sets which forms cliques in the chordal graph $G[ V \setminus Y]$. Since each bag in the clique-tree decomposition of a chordal graph correspond to its maximal cliques, 
%
%there exist nodes $a$, $b$ and $c$ in the tree of $\mathcal{T}$ such that $C_1 \setminus Y \subseteq B_a$,  $C_2 \setminus Y \subseteq B_b$ and $C_3 \setminus Y \subseteq B_c$ where $B_a, B_b$ and $B_c$ are the corresponding bags of the nodes $a, b$ and $c$ (Note that some of the nodes could be equal). From Lemma \ref{lemma:tree-sep}, we know that there exists a node $t$ of the tree such that $t$ separates the nodes $a$, $b$ and $c$. Hence the corresponding bag $B_t$ separates the sets $B_a, B_b$ and $B_c$ in the graph $G[V \setminus Y]$. Hence we have shown that there exist a set $X = B_t \cup Y$ where $G[B_t]$ forms a clique and $|Y| \leq k$ which separates the sets $C_1, C_2$ and $C_3$ in the graph $G$.

We compute a set $X'$ such that $G[X']$ is a $(1,k)$-semi clique and  $X'$ is a separator of  $C_i$ and $ C_j$ for all $i,j\in \{1,2,3\}$ and $i\neq j$, without the knowledge of a $(1,k)$-semi clique tree decomposition of $G$.
Let $X= X_1 \uplus X_2$ where $X_1$ is a clique and $|X_2| \leq k$. Using Lemma~\ref{lemma:maxicliq}, we either conclude that $G$ has no CVD of size $k$ or we go over all the maximal cliques of the graph $G$. We know that $X_1 \subseteq D$ for one of such maximal cliques $D$.  Now in the graph $G[V \setminus D]$, we know that there exists a set $Z \subseteq X_2 = X \setminus X_1 $ of size at most $k$ which separates the cliques $C_x \setminus D, C_y \setminus D$ and $C_z \setminus D$.
To find $Z$, we add three new vertices $x'$, $y'$ and $z'$. We make $x'$ adjacent to all the vertices of $C_x \setminus D$, $y'$ adjacent to all the vertices of $C_y \setminus D$ and $z'$ adjacent to all the vertices of $C_z \setminus D$. %  in the graph $G[W \setminus D]$.
 We find the node multiway cut $Y$ of size at most $k$ with the terminal set being $\{x', y', z'\}$. The set $Y$ can be found in $\OO(2^k km)$ using the known algorithm for node multiway cut \cite{DBLP:journals/toct/CyganPPW13,iwata20180}. If the algorithm returns that there is no such set $Y$ of size $k$, we conclude that there is no CVD of size at most $k$ in $G$. Else we get a set $X' = D \uplus Y$ which satisfies the properties of $X$. % which can be done in $\OO^*(2^k)$ time.

Now we estimate the time taken to obtain $X'$. We get all the $\OO(2^k \cdot n)$ maximal cliques of the graph in $\OO(2^k \cdot n^{\omega+1})$ time. Now for each maximal clique we use the $\OO(2^k km)$ algorithm for node multiway cut. Thus, the overall running time is $\OO(2^k \cdot n^{\omega+1} + (2^k n) \cdot (2^k km)) = \OO(4^{k} \cdot (kn^{3} +n^{\omega+1}) )$.
%$\OO(4^{k} \cdot n^{\omega+4})$.
\end{proof}

Now  we prove our main result (i.e., Theorem~\ref{theorem:treedecomp}) in this section. For convenience we restate it here.  

\begin{thmn} [\ref{theorem:treedecomp}]
There is an algorithm that given a graph $G$ and an integer $k$ runs in time $\OO(2^{7k} \cdot (kn^{4} +n^{\omega+2}))$ and either constructs a $(4,7k+5)$-semi clique tree decomposition $\mathcal{T}$ of $G$ or concludes that there is no chordal vertex deletion set of size $k$ in $G$. Moreover, the algorithm also provides a partition $C_1\uplus C_2\uplus C_3\uplus C_4 \uplus N$ of each bag of ${\cal T}$ such that  $\vert N\vert \leq 7k+5$ and $C_i$ is a clique in $G$ for all $i\in \{1,2,3,4\}$. 

%Given a graph $G$ such that it contains a chordal vertex deletion set of size $k$, there is an algorithm to construct a $(4, 7k+5)$ semi clique-tree decomposition $\mathcal{T}$ of $G$ in $\OO^*(2^{7k})$ time.
\end{thmn}
 %such that each bag  $X$ of the decomposition $\mathcal{T}$ is that $ X = C \cup N$ where $G[C]$ is the union of at most three maximal cliques and $|N| = O(k)$.
\begin{proof}

We assume that $G$ is connected as if not we can construct a $(4,7k+5)$-semi clique tree decomposition for each connected components of $G$ and attach all of them to a root node whose bag is empty to get the required $(4,7k+5)$-semi clique tree decomposition of $G$.

To construct a $(4,7k+5)$-semi clique tree decomposition $\mathcal{T}$, we define a recursive procedure ${\sf Decompose}(W,S,d)$ where $S \subset W \subseteq V(G)$ and $d \in \{0,1,2\}$. The procedure returns a rooted $(4, 7k+5)$-semi clique tree decomposition of $G[W]$ such that $S$ is contained in the root bag of the tree decomposition. The procedure works under the assumption that the following invariants are satisfied. 
\begin{itemize}
\item $G[S]$ is a $(d,6k+4)$-semi clique and $W \setminus S \neq \emptyset$.% $S = C \uplus N$ where $G[C]$ is the union of $d$ maximal cliques and $|N| \leq 6k+4$.
%\item  Both $G[W]$ and $G[W \setminus S]$ are connected.
%\item both $G[W]$ and $G[W \setminus S]$ are connected.
\item $S = N_G(W \setminus S)$. Hence $S$ is called the \textit{boundary} of the graph $G[W]$.%boundary definition
\end{itemize}

To get the required $(4, 7k+5)$-semi clique tree decomposition of $G$, we call ${\sf Decompose}(V(G),$  $\emptyset, 0)$ which satisfies all the above invariants.
The procedure ${\sf Decompose}(W,S,d)$ calls procedures ${\sf Decompose}(W',S',d')$ and a new procedure ${\sf SplitCliques}(W',S')$ whenever $d=2$. %$G[S']$ is a $(3,6k+4)$-semi clique and not a $(d,6k+4)$-semi clique  for any $d\in \{0,1,2\}$.
For these subprocedures, we will show that $|W' \setminus S'| < |W \setminus S|$. Hence by induction on cardinality of $W \setminus S$, we will show the correctness of the ${\sf Decompose}$ procedure.
%to some of these subproblems. 

The procedure ${\sf SplitCliques}(W,S)$ with $S \subset W \subseteq V(G)$ %where $S = C \uplus N$ with $G[S]$ is a $(3,5k+2)$ semi clique 
 also outputs a rooted $(4, 7k+5)$-semi clique tree decomposition of $G[W]$ such that $S$ is contained in the root bag of the tree decomposition. But the invariants under which it works are slightly different which we list below. % which The procedure works under the assumption that the following invariants are satisfied.  %being the union of exactly three cliques and $|N| \leq 5k+2$. 
\begin{itemize}
\item $G[S]$ is a $(3,5k+3)$-semi clique and $ W \setminus S \neq \emptyset$.
%\item  Both $G[W]$ and $G[W \setminus S]$ are connected.
\item $S = N_G( W \setminus S)$.
\end{itemize}

%Add invariants are satisfied for Decompose in the above case.
Notice that the only difference between invariants for ${\sf Decompose}$ and ${\sf SplitCliques}$ is the first invariant where we require $G[S]$ to be a $(3,5k+3)$-semi clique for ${\sf SplitCliques}$ and $(d,6k+4)$-semi clique for ${\sf Decompose}$.
% 
%Let $d$ be the number of cliques whose union is $C$. We have $d \leq 2$. 

The procedure ${\sf SplitCliques}(W,S)$ calls procedures ${\sf Decompose}(W',S',2)$ where we will again show that $|W' \setminus S'| < |W \setminus S|$. Hence again by induction on cardinality of $W \setminus S$, we will show the correctness. % of the ${\sf Decompose}$ procedure.
Now we describe how the procedure ${\sf Decompose}$ is implemented. 

\medskip
\noindent
\textbf{Implementation of ${\sf Decompose}(W,S,d)$:}
Notice that $d\in \{0,1,2\}$. 
Firstly, if $|W \setminus S| \leq k+1$, we output the tree decomposition as a node $r$ with bag $X_r=W$ and stop. Clearly the graph $G[X_r]$ is a $(4,7k+5)$-semi clique and it contains $S$. Otherwise, we do the following.

We construct a set $\hat{S}$ with the following properties.
\begin{enumerate}%[label = \alph*]
\item $S \subset \hat{S} \subseteq W \subseteq V(G)$.
\item $G[\hat{S}]$ is a $(d+1, 7k+5)$-semi clique. Let $\hat{S} = C' \uplus N'$ where $G[C']$ is the union of $d+1$ cliques and $|N'| \leq 7k+5$.
\item Every connected component of $G[W \setminus \hat{S}]$ is adjacent to at most $5k+3$ vertices of $N'$.
\end{enumerate}

Since $G[S]$ is a $(d,6k+4)$-semi clique, we have that $S = C \uplus N$, where $G[C]$ is the union of $d$ cliques and $|N| \leq 6k+4$.  

\medskip
\noindent
{\bf Case 1: $|N| < 5k+3$.} We set $\hat{S}=S\cup \{u\}$, where $u$ is an arbitrary vertex in $W \setminus S$. Note that this is possible as $W \setminus S \neq \emptyset$. Clearly $\hat{S}$ follows all the properties above.

\medskip
\noindent
{\bf Case 2: $5k+3 \leq |N| \leq 6k+4$.} Note that $G[W]$ being a subgraph of $G$ also has a chordal vertex deletion set of size at most $k$ if $G$ has it. Applying Corollary \ref{corollary:separation4k} for the graph $G[W]$ and the subset $N$, we either conclude that $G$ has no CVD of size $k$ or get a partition $(N_A, N_B)$ of $N$, a subset $X\subseteq W$ and a partition $D\uplus Z$ of $X$, where 
$D$ is a clique in $G[W]$ and $\vert Z\vert \leq k$,   in time $\OO(2^{7k} \cdot (kn^{3} +n^{\omega+1}) )$ such that $|N_A|, |N_B| \leq 4k+2$ and $X$ is a vertex separator of $N_A$ and $N_B$ in the graph $G[W]$.% and $G[X]$ is a $(1,k)$-semi clique.
%Also given the 2-partition $(N_A, N_B)$, we can find a subset $X$ satisfying the above properties in $\OO^*(2^k)$ time. 
%\end{itemize}

%%Let us define a weight function $w: W \rightarrow \mathbb{R}_{\ge 0}$ such that $w(v)=1$ if $v \in N$ and $0$ otherwise. From Lemma \ref{lemma:balsep}, we know that there exists a set $X$ which is the balanced separation of $G[W]$ with respect to $w$ where $G[X]$ is an $(1,k)$ semi clique.
%
%Let $X = C'' \uplus N''$ where $C''$ is a clique and $|N''| \leq k$. %Also let $(A,B)$ be the corresponding 2-partition of $W \setminus X$.
%We use Lemma~\ref{lemma:maxicliq} to go over all maximal cliques of $G[W]$ to find a maximal clique $D$ such that $C'' \subseteq D$. %(the clique part of $X$) is a subset of $D$.
% We can conclude that in the remaining graph $G[W \setminus D]$, there exists a separator $Z \subseteq N'' = X \setminus C''$ of size at most $k$ for the sets $N_A$ and $N_B$. 
%
%%Define $S_A$ as $A \cap S$ and $S_B$ as $B \cap S$. Also define $N_A = S_A \cap N$  and $N_B = S_B \cap N$. Since $X$ is a balanced separation with respect to $w$, we have $|N_A|, |N_B| \leq \frac{2}{3}|N| \leq \frac{2}{3}(6k+4) \leq 4k+2$. %Also $|N_A|, |N_B| \geq |N| - \frac{2}{3}|N| = 2k+2$.
%
%We go over all $2^{|N|} \leq 2^{6k+4}$ 2-partitions of $N$ to guess the partition $(N_A,N_B)$. Then we apply the classic Ford-Fulkerson maximum flow algorithm to find the separator $Z$ of the sets $N_A$ and $N_B$ in the graph $G[W \setminus D]$. %By the existence of $N''$, we know that $|Z| \leq k$.

%Let $Y = D \uplus Z$. 
We define $\hat{S} = S \cup X \cup \{u\}$ where $u$ is an arbitrary vertex in $W \setminus S$. 
We need to verify that $\hat{S}$ satisfies the required properties. 

%Clearly $G[\hat{S}]$ is an $(d+1,7k+5)$ semi clique.
\begin{claim} The set $\hat{S}$ satisfies properties $(1), (2)$ and $(3)$.
\end{claim}
\begin{proof}

Since $u \in W \setminus S$, $S \subset \hat{S}$. Hence $\hat{S}$ satisfies property $(1)$.

We now show that $\hat{S}$ satisfies property $(2)$. 
Recall that $S = C \uplus N$ , where $G[C]$ is the union of $d$ cliques and $|N| \leq 6k+4$.  
We define  sets $C' = C \cup D$ and $N' = ((N \cup Z)  \setminus C') \cup \{u\}$
%((D \cap N) \cup (C \cap Z))$. 
%We have $\hat{S} = S \cup Y = (C \uplus N) \cup (D \uplus Z) = (C \cup D) \cup (N \cup Z) = (C \cup D) \uplus \big((N \cup Z) \setminus ((D \cap N) \cup (C \cap Z)\big)= C' \uplus N'$. 
%Let $\hat{S} = C' \cup N'$ where  
Notice that $\hat{S} = C' \cup N'$. 
Clearly $G[C']$ is the union of $d+1$ cliques. Also $|N'| \leq |N| + |Z| + 1 \leq (6k+4)+k + 1 \leq 7k+5$.
Thus $\hat{S}$ satisfies property (2).

We now show that $\hat{S}$ satisfies property $(3)$. 
Recall $\hat{S} = C' \cup N'$, where $C'=C\cup D$ and $N'=((N\cup Z)\setminus C') \cup \{u\}$. 
Recall that $X=D\cup Z \subseteq \hat{S}$ is separator of $N_A$ and $N_B$.  
where $N=N_A\uplus N_B$ and  $|N_A|, |N_B| \leq 4k+2$. This implies that any connected component $H$ in $G[W\setminus X]$ can contain at most $4k+2$ vertices from $N$ as the neighborhood of $V(H)$ is contained in $X$, because $X$ is a separator. Moreover $\vert Z\vert \leq k$. This implies that any connected component in $G[W\setminus \hat{S}]$ is adjacent to at most $4k+2$ vertices in $N$ and at most $k$ vertices in $Z$, and hence at most $5k+3$ vertices in $N'=((N\cup Z)\setminus C') \cup \{u\}$. 
\end{proof}
Now we define the recursive subproblems arising in the procedure ${\sf Decompose}$ $(W,S,d)$ using the constructed set $\hat{S}$. If $\hat{S}=W$, then there will not be any recursive subproblem. 
Otherwise, let $P_1, P_2, \dotsc, P_q$ be vertex sets of the connected components of $G[W \setminus \hat{S}]$ and 
$q\geq 1$ because  $\hat{S}\neq W$. 
We have the following cases:
%\begin{itemize}
%\item 

\medskip
\noindent
{\bf Case 1: $d < 2$: }
For each $i \in [q]$, recursively call the procedure ${\sf Decompose}(W'= N_G[P_i], S'= N_G(P_i), d+1)$.

We now show that the  invariants are satisfied for procedures ${\sf Decompose}(W' = N_G[P_i],S' = N_G(P_i),d+1)$ for all $i \in [q]$.
Let $Q_i = S' \cap N'$. Note that from condition $(3)$ for $\hat{S}$, we have $|Q_i| \leq 5k+3 $. Since $S' \setminus Q_i \subseteq C'$ and $G[C']$ is a union of $d+1$ cliques, $G[S']$ forms a $(d+1,5k+3)$-semi clique which is also a $(d+1,6k+4)$-semi clique.  Also by definition of neighbourhoods, $P_i = N_G[P_i] \setminus N_G(P_i) = W' \setminus S'$. Since $P_i$ is a non-empty set by definition, $W' \setminus S'$ is non-empty. Hence the first invariant required for the ${\sf Decompose}$ is satisfied. 
%
%We now show that the second invariant is satisfied.  We have $P_i = N_G[P_i] \setminus N_G(P_i)$ and $G[P_i]$ is connected. Now we establish that the graph $G[N_G[P_i]]$ is connected as well. Look at  a pair of vertices $x,y \in  N_G(P_i)$. Clearly it has respective neighbours $x',y' \in P_i$. Since $G[P_i]$ is connected, there is a path from $x'$ to $y'$ which also connects $x$ and $y$.
%
Since $S' = N_G(P_i) = N_G(N_G[P_i] \setminus N_G(P_i)) = N_G(W' \setminus S')$, the second invariant is satisfied. %Other two variants can be easily verified.

%\item Case 2: $d+1 =3$:

\medskip
\noindent
{\bf Case 2: $d = 2$: }
For each $i \in [q]$, recursively call the procedure ${\sf SplitCliques}(W'= N_G[P_i], S'= N_G(P_i))$. We can show that the invariants for ${\sf SplitCliques}$ are satisfied with the proofs similar to previous case. %being the same as established above. % for recursive subproblems calling $Decompose$ procedure.   %The procedure $Decompose2(W,S)$ for $S \subset W \subseteq V(G)$ returns a rooted tree decomposition of $G[W]$ such that $S$ is contained in the root bag of the tree decomposition, but the invariants required are different here which we give below. %We will show that the invariants required for $Decompose2$ are satisfied later. 

%\end{itemize}

We now explain how to construct the $(4,7k+5)$-semi clique tree decomposition using ${\sf Decompose}(W,S,d)$. 
Here, we assume that ${\sf Decompose}(W',S',d+1)$ and ${\sf SplitCliques}(W',S')$ return a  $(4,7k+5)$-semi clique tree decomposition $G[W']$ 
when $\vert W'\setminus S'\vert <\vert W\setminus S\vert$. 
That is, we apply induction on $|W \setminus S|$. 
Look at the subprocedures ${\sf Decompose}(W',S',d)$ and ${\sf SplitCliques}(W',S')$. We have $W' \setminus S'  = N_G[P_i] \setminus N_G(P_i) = P_i$ which is a subset of $W \setminus \hat{S}$ which in turn is a strict subset of $W \setminus S$. Hence $| W' \setminus S'| < |W \setminus S|$. Hence we apply induction on $|W \setminus S|$ to the subprocedures. Let $ \mathcal{T}_i$ be the  $(4,7k+5)$-semi clique tree decomposition obtained from the subprocedure with $W' = N_G[P_i] $ and $S' = N_G(P_i) $. Let $r_i$ be the root of $ \mathcal{T}_i$ whose associated bag is $X_{r_i}$. By induction hypothesis $S'\subseteq X_{r_i}$. We create a node $r$ with the corresponding bag $X_r = \hat{S}$. For each $i \in [q]$, we attach $ \mathcal{T}_i$ to $r$ by adding edge $(r, r_i)$. Let us call the tree decomposition obtained so with root $r$ as $\mathcal{T}$. We return $\mathcal{T}$ as the output of ${\sf Decompose}(W,S,d)$. By construction, it easily follows that $\mathcal{T}$ is a $(4,7k+5)$-semi clique tree decomposition of the graph $G[W]$ with the root bag containing $S$. We note that when $W=\hat{S}$, the procedure returns a single node tree decomposition with $X_r=W=\hat{S}$.

\medskip
\noindent
\textbf{Implementation of ${\sf SplitCliques}$ Procedure: }
Again if $|W \setminus S| \leq k+1$, we output the tree decomposition as a node $r$ with bag $X_r=W$ and stop. Clearly the graph $G[X_r]$ is a $(4,7k+5)$ semi clique and it contains $S$. Otherwise we do the following.
Let $S = C \uplus N = (C_x \uplus C_y \uplus C_z) \uplus N$ where $C_x, C_y$ and $C_z$ are the vertex sets of the three cliques in $G[C]$. We apply Lemma \ref{lemma:3cliqsep} to graph $G[W]$ and sets $C_x, C_y$ and $C_z$, to either conclude that $G$ has no CVD of size $k$ or obtain a set $Y$ such that $Y$ separates the sets $C_x, C_y$ and $C_z$ and $G[Y]$ is a $(1,k)$-semi clique. Let $Y = D \uplus X$ where $D$ is a clique and $|X| \leq k$. 

Let $Y' = Y \cup \{u\}$ where $u$ is any arbitrary vertex from $W \setminus S$ which we know to be non-empty.
If $S\cup Y'=W$, then it will not call any recursive subproblem. 
Otherwise, let $P_1, P_2, \dotsc, P_q$ be the connected components of the graph $G[W \setminus (S \cup Y')]$. We recursively call ${\sf Decompose}(W' = N_G[P_i], S' = N_G(P_i),2)$ for all $i \in [q]$. %, $Decompose(N_G[P_y],$ $ N_G(P_y),2)$ and $Decompose(N_G[P_z], N_G(P_z),2)$.

Since $Y'$ is a separator of the cliques $C_x, C_y$ and $C_z$, any connected component $P_i$ will have neighbours to at most one of the three cliques $C_x \setminus Y', C_y \setminus Y'$ and $C_z \setminus Y'$ in $G[W \setminus (S \cup Y')]$.
% are cliques, they will belong to exactly one of these components. Since $Y'$ separates these sets, the all the components will be different as well. Let $P_x, P_y, P_z$ be vertex sets of the connected components of $G[W \setminus Y']$ containing $C_x \setminus Y', C_y \setminus Y'$ and $C_z \setminus Y'$. 
%
We show that the invariants required for the procedure ${\sf Decompose}$ is satisfied in these subproblems. Let us focus on the procedure ${\sf Decompose}(W' = N_G[P_i],S'= N_G(P_i),2)$ which has neighbours only to the set $C_x \setminus Y'$.
%Note that although its possible that the entire set $N$ is adjacent to some of the vertex sets, say $P_x$. 
We define sets $C' = C_x \cup D$ and $N'= (N \cup X \cup \{u\}) \setminus C'$. The vertex set  $P_i$ has neighbours only to the set $ (C_x \uplus N) \cup Y' = (C_x \uplus N) \cup (D \uplus X) \cup \{u\} = (C_x \cup D) \cup (N \cup X \cup \{u\}) = C' \uplus N'$. Clearly $G[C']$ is the union of at most two cliques and $|N'| \leq |N| + |X| + 1 = 5k+3 +k + 1 \leq 6k+4$. 
%Other invariants?
Hence the first invariant is satisfied for the procedure ${\sf Decompose}(N_G[P_i], N_G(P_i),2)$. 
The proof of the second invariant is the same as to that of the subproblems of ${\sf Decompose}$ procedure.
The satisfiability of invariants for other subprocedures can also be proven similarly.

We now construct the $(4,7k+5)$-semi clique tree decomposition  returned by ${\sf SplitCliques}$ $(W,S)$. Again we apply induction on $|W \setminus S|$. Consider the subprocedures ${\sf Decompose}(W',S',d)$. We have $W' \setminus S'  = N_G[P_i] \setminus N_G(P_i) = P_i$ which is a subset of $W \setminus (S \cup Y')$ which in turn is a strict subset of $W \setminus S$ as $u \in W \setminus S$ is present in $Y'$. Hence $| W' \setminus S'| < |W \setminus S|$ and we apply induction on $|W \setminus S|$ to the subprocedures. Let $ \mathcal{T}_i$ be the  $(4,7k+5)$-semi clique tree decomposition obtained from the subprocedure with $W' = N_G[P_i] $ and $S' = N_G(P_i) $. Let $r_i$ be the root of $ \mathcal{T}_i$ whose bag $X_{r_i}$ we show contains $S'$. We create a node $r$ with the corresponding bag $X_r = S \cup Y' = (C_x \uplus
C_y \uplus C_z \uplus D) \uplus N'$.
For each $i \in [q]$, we attach $ \mathcal{T}_i$ to $r$ by adding edge $(r, r_i)$. Let us call the tree decomposition obtained so with root $r$ as $\mathcal{T}$. We return $\mathcal{T}$ as the output of ${\sf SplitCliques}(W,S,d)$. By construction, it easily follows that $\mathcal{T}$ is a $(4,7k+5)$-semi clique tree decomposition of the graph $G[W]$ with the root bag containing $S$. 
We mention that when $W=S\cup Y'$, the procedure returns a single node tree decomposition with $X_r=W$.

\medskip
\noindent
\textbf{Running time analysis: }
In the procedure ${\sf Decompose}$, we invoke Corollary \ref{corollary:separation4k} which takes $\OO(2^{7k} \cdot (kn^{3} +n^{\omega+1}))$ time.
%we first go over all $\OO(2^k \cdot n)$ maximal cliques of the graph which takes $\OO(2^k \cdot n^{\omega+1})$ time. Then we go over at most $2^{6k+4}$ guesses for $N_A$ and $N_B$.  Finally we use the Ford-Fulkerson maximum flow algorithm to find the separator of size at most $k$ for $N_A$ and $N_B$ which takes $\OO(k(n+m))$ time. %\cite{cygan2015parameterized}.
% Overall the running time is $\OO(2^{6k} \big( (2^k \cdot n^{\omega+1}) + (k(n+m)) \big) ) = \OO(2^{7k} \cdot n^{\omega+1})$.
 For the procedure ${\sf SplitCliques}$, we invoke Lemma \ref{lemma:3cliqsep} which takes $\OO(4^{k} \cdot (kn^{3} +n^{\omega+1}))$ time.
%we go over all the $\OO(2^k \cdot n)$ maximal cliques of the graph in $\OO(2^k \cdot n^{\omega+1})$ time and use the $\OO(2^k km)$ algorithm for node multiway cut. Overall running time is $\OO(4^{k} \cdot n^{\omega+1})$.
All that is left is to bound the number of calls of the procedures ${\sf Decompose}$ and ${\sf SplitCliques}$. %Since ${\sf SplitCliques}$ is always called as subroutines of ${\sf Decompose}$ and itself only calls ${\sf Decompose}$ as subroutines, bounding the number of calls of ${\sf Decompose}$ is enough. 
Each time ${\sf Decompose}$ or ${\sf SplitCliques}$ is called, it creates a set $\hat{S}$ (in the case of ${\sf SplitCliques}$, $\hat{S}=S\cup Y'$) which is a strict superset of $S$. This allows us to  map each call of ${\sf Decompose}$ or ${\sf SplitCliques}$ to a unique vertex $u \in \hat{S} \setminus S$ of $V(G)$. Hence the total number of calls of ${\sf Decompose}$ and ${\sf SplitCliques}$ is not more than the total number of vertices $n$. %Explained
 Hence the overall running time of the algorithm which constructs  the $(4,7k+5)$-semi clique tree decomposition of $G$ is $\OO(2^{7k} \cdot (kn^{4} +n^{\omega+2}))$.
\end{proof}

 %The algorithm when invoking Lemma \ref{corollary:separation4k} for some  procedure ${\sf Decompose}(W,S,d)$ can detect that there is no chordal vertex deletion set of size $k$ in $G$ in two cases. The first case is when it enumerates the maximal cliques of $G$. If the number of maximal cliques enumerated is more than $\OO(2^k \cdot n)$, by Lemma \ref{lemma:maxicliq}, we know our assumption is false. The second case is when the separator $X$ returned by the lemma has size more than $k$. The algorithm also can detect the assumption is false when invoking Lemma \ref{lemma:3cliqsep} for some procedure ${\sf SplitCliques}(W,S)$ when the separator returned by the lemma has size more than $k$. We have the following corollary.

% Whenever algorithm detects that $G$ does not have a CVD of size $k$, it exits. It detects so when it finds that the number of maximal cliques enumerated is more than $\OO(2^k \cdot n)$ and when the separator computed is of size more than $k$. We have the following corollary.

%\begin{corollaryn}[\ref{corollary:treedecomp}]
%There is an algorithm that given a graph $G$ and an integer $k$, runs in time $\OO(2^{7k} \cdot (kn^{4} +n^{\omega+2}))$ and either constructs a $(4,7k+5)$-semi clique tree decomposition $\mathcal{T}$ of $G$ or concludes that there is no chordal vertex deletion set of size $k$ in $G$. Moreover, the algorithm also provides a partition $C_1\uplus C_2\uplus C_3\uplus C_4 \uplus N$ of each bag of ${\cal T}$ such that  $\vert N\vert \leq 7k+5$ and $C_i$ is a clique in $G$ for all $i\in \{1,2,3,4\}$. 
%\end{corollaryn}

%!TEX root = main_lipics.tex

\section{Structural Parameterizations with Chordal Vertex Deletion Set} \label{sec:Structural Param}

%%%%Vibha
%We first give the known folklore algorithm for \VCCVD with the modulator given as input for fullness.
%\begin{theorem}
%\VCCVD when given a subset $S \subseteq V(G), |S| \leq k$ such that $G - S$ is a chordal graph is also given as input, can be solved in $\OO^*(2^{k})$ time.
%\end{theorem}
%\begin{proof}
%The algorithm first make a guess $S' \subseteq S$ of the vertex cover solution $X$ intersecting with $S$. Now we have the graph $G - S'$ where we can only add vertices from $V \setminus S$ to the solution. If there is an edge in $G - S'$ such that one of its endpoints is in $S \setminus S'$, then the other endpoint in $V \setminus S$ must go in the solution to cover this edge. Identify all such edges and add the corresponding vertices in $ V \setminus S$ to $X$. If at any point, the budget is exceeded, we return NO.
%Now we have that all the edges remaining to be covered are such that both the endpoints are in $ V \setminus S$. Since $G[V \setminus S]$ is a chordal graph, the remaining edges form a chordal graph as well. We use the known polynomial time algorithm to optimally solve {\sc Vertex Cover} in this chordal graph~\cite{golumbic2004algorithmic}.
%The leading factor in the running time is the number of guesses for $S'$ which is at most $2^k$. Hence we have an algorithm with $\OO^*(2^{k})$ running time.
%\end{proof}

Now, we briefly explain a DP algorithm using semi clique tree decomposition to prove the following theorem. %Theorem~\ref{theorem:fpt}.

%\begin{thmn} [\ref{theorem:fpt}]
\begin{theorem} \label{theorem:fpt}
There is a $\OO(2^{7k} n^5)$ time algorithm for \VCCVD\ that either returns minimum vertex cover of $G$ or concludes that there is no CVD of size $k$ in $G$.
%There is an algorithm that given a graph $G$ and an integer $k$ runs in time $\OO(2^{7k} \cdot (kn^{4} +n^{\omega+2}))$ and either constructs a $(4,7k+5)$-semi clique tree decomposition $\mathcal{T}$ of $G$ or concludes that there is no chordal vertex deletion set of size $k$ in $G$. Moreover, the algorithm also provides a partition $C_1\uplus C_2\uplus C_3\uplus C_4 \uplus N$ of each bag of ${\cal T}$ such that  $\vert N\vert \leq 7k+5$ and $C_i$ is a clique in $G$ for all $i\in \{1,2,3,4\}$. 
\end{theorem}

%In the following DP's, using Theorem \ref{theorem:treedecomp} and Lemma \ref{lemma:nice-tree} we get $\mathcal{T}=\{T,X_t\}$ as a nice $(3,7k+5)$ semi clique-tree decomposition of $G$ in $\OO^*(2^{7k})$ time. Consider $r$  to be the root. For each node $t \in V(T)$, we denote $X_t$ to be the vertices of $G$ contained in bag $t$. For any vertex $t \in V(T)$, we call $D_t$ to be the set of vertices that are descendant of $t$.
%We define $H_t$ to be the subgraph of $G$ induced on the vertex set $X_t \cup \bigcup_{t' \in D_t} X_{t'} $ except for the edges in $G[X_t]$. Let $X_t=C_t \cup N_t$ for all $t \in V(T)$ where $C_t$ contains at most three cliques and $|N_t| \leq 7k+4$. With these terminologies in mind, we now give an FPT algorithm for \VCCVD. 

\begin{proof}[Proof sketch]
First, we use Theorem \ref{theorem:treedecomp} to construct a $(4,7k+5)$-semi clique tree decomposition $\mathcal{T}=(T,\{X_t\}_{t\in V(T)})$ of $G$ in $\OO^*(2^{7k})$ time. 
In the tree decomposition  $\mathcal{T}=(T,\{X_t\}_{t\in V(T)})$, for any vertex $t \in V(T)$, we call $D_t$ to be the set of vertices that are descendant of $t$. We define $G_t$ to be the subgraph of $G$ on the vertex set $X_t \cup \bigcup_{t' \in D_t} X_{t'}$. 
%From Lemma \ref{lemma:nice-tree}, we can assume that the $\mathcal{T}$ is a nice $(3,7k+4))$ semi clique-tree decomposition.
We briefly explain the DP table entries on $\mathcal{T}$. 
%and prove that the number of table entries we need to consider is at most $\OO(2^{7k} n^5)$. 
Arbitrarily root the tree  $T$ at a node $r$.   
%For any vertex $t \in V(T)$, we call $D_t$ to be the set of vertices that are descendant of $t$. We define $G_t$ to be the subgraph of $G$ induced on the vertex set $X_t \cup \bigcup_{t' \in D_t} X_{t'}$.  
Let $X_t=C_{t,1}\uplus\ldots \uplus C_{t,4} \uplus N_t$ where $|N_t| \leq 7k+5$ and $C_{t,j}$ is a clique in $G$ for all $j\in \{1,\ldots,4\}$. In a standard DP for each node $t\in V(T)$ and $Y\subseteq X_t$, we have a table entry $DP[Y,t]$ which stores the value of a minimum vertex cover $S$ of $G_t$ such that $Y=X_t\cap S$ and if no such vertex cover exists, then  $DP[Y,t]$ stores $\infty$. In fact we only need to store $DP[Y,t]$ whenever it is not equal to $\infty$. 
Now consider a bag $X_t$ in ${\cal T}$. For any $Y\subseteq X_t$, if $\vert C_{t,j} \setminus Y \vert \geq 2$ for any $j\in [4]$, then 
 $DP[Y,t]=\infty$ because $C_{t,j}$ is a clique. Therefore, we only need to consider subsets $Y\subseteq X_t$ for which $\vert C_{t_j}\setminus Y\vert \leq 1$ for all $j\in [4]$. The number of choices of such subsets $Y$ is bounded 
 by $\OO(2^{7k}n^4)$. This implies that the total number of DP table entries is $\OO(2^{7k}n^5)$. All these values can be computed in time $\OO(2^{7k}n^5)$ time using standard dynamic programming in a bottom up fashion. For more details about dynamic programming over tree decomposition, see ~\cite{cygan2015parameterized}.
\end{proof}

In a similar way, using the fact that any odd cycle transversal or feedback vertex set contains all but at most two vertices from each clique, we can give FPT algorithms for  following theorems. 

\begin{theorem} \label{theorem:fptfvs}
There is an algorithm for {\sc FVS By CVD} running in time $2^{\OO(k)} n^{\OO(1)}$ that either returns minimum feedback vertex set of $G$ or concludes that there is no CVD of size $k$ in $G$.
\end{theorem}
\begin{proof}[Proof sketch]
%\subsection{Proof sketch of Theorem \ref{theorem:fptfvs}}
We use the ideas from the DP algorithm for {\sc Feedback Vertex Set} using the rank based approach~\cite{bodlaender2015deterministic}.
We create an auxiliary graph $G'$ by adding a vertex $v_0$ to $G$ and making it adjacent to all the vertices of $G$. Let $E_0$ be the set of newly added edges.
We again use Theorem \ref{theorem:treedecomp} to construct a $(4,7k+5)$-semi clique tree decomposition of $G$ and add $v_0$ to all the bags to get the tree decomposition $\mathcal{T}=(T,\{X_t\}_{t\in V(T)})$ of $G'$ in $\OO^*(2^{7k})$ time. We use a dynamic programming algorithm for {\sc Feedback Vertex Set} on $\mathcal{T}$ where the number of entries of the DP table we will show to be $2^{7k+5}n^{11}$. Let $X_t=C_{t,1}\uplus\ldots \uplus C_{t,4} \uplus N_t$ for all $t \in V(T)$ where $|N_t| \leq 7k+5$ and $C_{t,j}$ is a clique in $G$ for all $j\in \{1,\ldots,4\}$. For a node $t \in V(T)$, a subset $Y \subseteq X_t$ and integers $i,j \in [n]$, we define the entry $DP[t,Y,i,j]$. The entry $DP[t,Y,i,j]$ stores a partition $\mathcal{P}$ of $Y$ if
\begin{itemize}
\item there exists a vertex subset $X \subseteq D_t$ , $v_o \in X$ such that $X \cap X_t = Y$ and
\item there exists an edge subset $X_0 \subseteq E(G_t) \cap E_0$ such that in the graph $(X, E(G_t[X \setminus \{v_0\}]) \cup X_0)$, we have $i$ vertices, $j$ edges, no connected component is fully contained in $D_t \setminus X_t$ and the elements of $Y$ are connected according to the partition $\mathcal{P}$.
\end{itemize}
We set $DP[t,Y,i,j]=\infty$ if the entry can be inferred to be invalid from $Y$.

We claim that {\sc Feedback Vertex Set by CVD} is a yes instance if and only if for the root $r$ of $\mathcal{T}$ with $X_r = \{v_0\}$ and some $i \geq |V| - \ell$, we have $DP[r,\{v_0\},i,i-1]$ to be non-empty. In the forward direction, we have a feedback vertex set $Z$ of size $\ell$. The graph $G - Z$ has $|V| - \ell$ vertices and $|V| - \ell - c$ edges where $c$ is the number of connected components of $G - Z$. We define $X = V \setminus Z \cup \{v_0\}$ and $X_0$ to be $c$ edges connecting $v_0$ to any one of the vertices of each of the $c$ components of $V \setminus Z$. We have $|X| \geq |V| - \ell$. The graph $(X, E(G_t[X \setminus \{v_0\}]) \cup X_0)$ has $|V| - \ell$ edges and satisfies the properties required for an entry in $DP[r,\{v_0\},i,i-1]$. In the reverse direction, we have a graph $(X, E(G_t[X \setminus \{v_0\}]) \cup X_0)$ having $i$ edges and $i-1$ edges. Since no connected component of the graph can be contained in $V(G_t) \setminus \{v_0\}$, the graph is a tree. Hence $V \setminus X$ is a feedback vertex set.

Now consider a bag $X_t$ in ${\cal T}$. For any $Y\subseteq X_t$, if $\vert C_{t,j} \setminus Y \vert \geq 3$ for any $j\in [4]$, then  $DP[t,Y,i,j]=\infty$ because $C_{t,j}$ is a clique. Therefore, we only need to consider subsets $Y\subseteq X_t$ for which $\vert C_{t_j}\setminus Y\vert \leq 2$ for all $j\in [4]$. The number of choices of such subsets $Y$ is bounded by $\OO(2^{7k}n^8)$. This implies that the total number of DP table entries is $\OO(2^{7k}n^{11})$. In each DP table entry $DP[t,Y,i,j]$, we store partitions of $Y$. The cardinality of $Y$ is bounded by $7k+13$ as $\vert C_{t_j}\setminus Y\vert \leq 2$ for all $j\in [4]$. Hence the number of entries stored in $DP[t,Y,i,j]$ can be bounded to be $2^{7k+13}$. 

The recurrences for computing $DP[t,Y,i,j]$ remains the same as in~\cite{bodlaender2015deterministic}. 
%It can be shown that the time taken to compute $DP[t,Y,i,j]$ is bounded by the time taken at join node.
 Using the ideas from \cite{bodlaender2015deterministic}, the time taken to compute all the  table entries of a particular node $t$ can be shown to be $ \OO((1+ 2^{\omega+1})^{7k+13} \cdot (7k+13)^{\OO(1)} \cdot n^{11})$. Taking the number of nodes to be $m=\OO(n^2)$ in the worst case, the overall running time is $ \OO((1+ 2^{\omega+1})^{7k+13} \cdot (7k+13)^{\OO(1)} \cdot n^{13})$.
 \end{proof}
\begin{theorem} \label{theorem:fptoct}
There is an algorithm for {\sc OCT By CVD} running in time $2^{\OO(k)} n^{\OO(1)}$ that either returns minimum odd cycle transversal of $G$ or concludes that there is no CVD of size $k$ in $G$.
\end{theorem}
\begin{proof}[Proof sketch]
%Next consider {\sc Odd Cycle Transversal by CVD}. 
Let $\mathcal{T}=(T,\{X_t\}_{t\in V(T)})$ be a $(4,7k+5)$-semi clique tree decomposition of the input graph $G$.  For each node $t\in V(T)$ and sets $Y_1, Y_2 \subseteq X_t$, we have a table entry $DP[Y_1, Y_2, Y_3=X_t \setminus (Y_1 \cup Y_2),t]$ which stores the value of a minimum odd cycle transversal $S$ of $G_t$ such that $Y_3=X_t\cap S$ and $(Y_1, Y_2)$ is a bipartition of $X_t \setminus Y_3$ which extends to a bipartition of $G_t \setminus S$.% and if no such vertex cover exists, then  $DP[Y,t]$ stores $\infty$.

For any $t\in V(T)$, let $X_t=C_{t,1}\uplus\ldots \uplus C_{t,4} \uplus N_t$ where $|N_t| \leq 7k+5$ and $C_{t,j}$ is a clique in $G$ for all $j\in \{1,\ldots,4\}$. Then, any odd cycle transversal contains all but at most two vertices from each clique $C_{1,j}$, $i\in [4]$. Using this fact we can bound the number of DP table entries to be at most $2^{\OO(k)}n^{\OO(1)}$. Then, we have the theorem from the standard dynamic programming for odd cycle transversal on tree decompositions. %get Theorem~\ref{theorem:fptoct}. 
\end{proof}

\subsection{SETH Lower Bounds}

%Next we define a problem  called {\sc Vertex Cover by ClsVD}. 

We give a $\OO^*((2- \epsilon)^{k})$ lower bounds for \VCCVD\ , 
%(defined in preliminaries) 
 {\sc FVS by CVD} and {\sc OCT by CVD}  assuming the Strong Exponential Time Hypothesis(SETH). 
 %Since cluster graph is a subclass of chordal graphs, deletion distance to a chordal graph will be an even smaller parameter. Hence the lower bound also holds for \VCCVD. 

\begin{theorem} \label{thm:seth-vc}
\VCCVD\ cannot be solved in $\OO^*((2- \epsilon)^{k})$ time for any $\epsilon >0$ assuming SETH.
\end{theorem}
\begin{proof}
We give a reduction from {\sc Hitting Set} defined as follows.

\noindent{\sc Hitting Set} : In any instance of {\sc Hitting Set}, we are given a set of elements $U$ with $|U|=n$, a family of subsets $\mathcal{F}=\{F \subseteq U\}$ and a natural number $k$. The objective is to find a set $S \subseteq U$, $|S| \leq k$ such that $S \cap F \neq \emptyset$ for all $F \in \mathcal{F}$.

The problem cannot be solved in $\OO^*((2- \epsilon)^{n})$ time assuming SETH ~\cite{DBLP:journals/talg/CyganDLMNOPSW16}.

Consider a {\sc Hitting Set} instance $(U,\FF)$. We construct an instance of {\sc Vertex Cover by ClsVD} as follows.
For each element $u \in U$, we add a vertex $v_u$. For each set $S \in \FF$, we add $|S|$ vertices corresponding to the elements in $S$. We also make the vertices of $S$ into a clique. Finally, for each element $u \in U$, we add edges from $v_u$ to the vertex corresponding to $u$ for each set in $S$ that contains $u$. See Figure \label{figpoint}.

\begin{figure}[!ht]
\begin{center}
\includegraphics[width=0.8\textwidth]{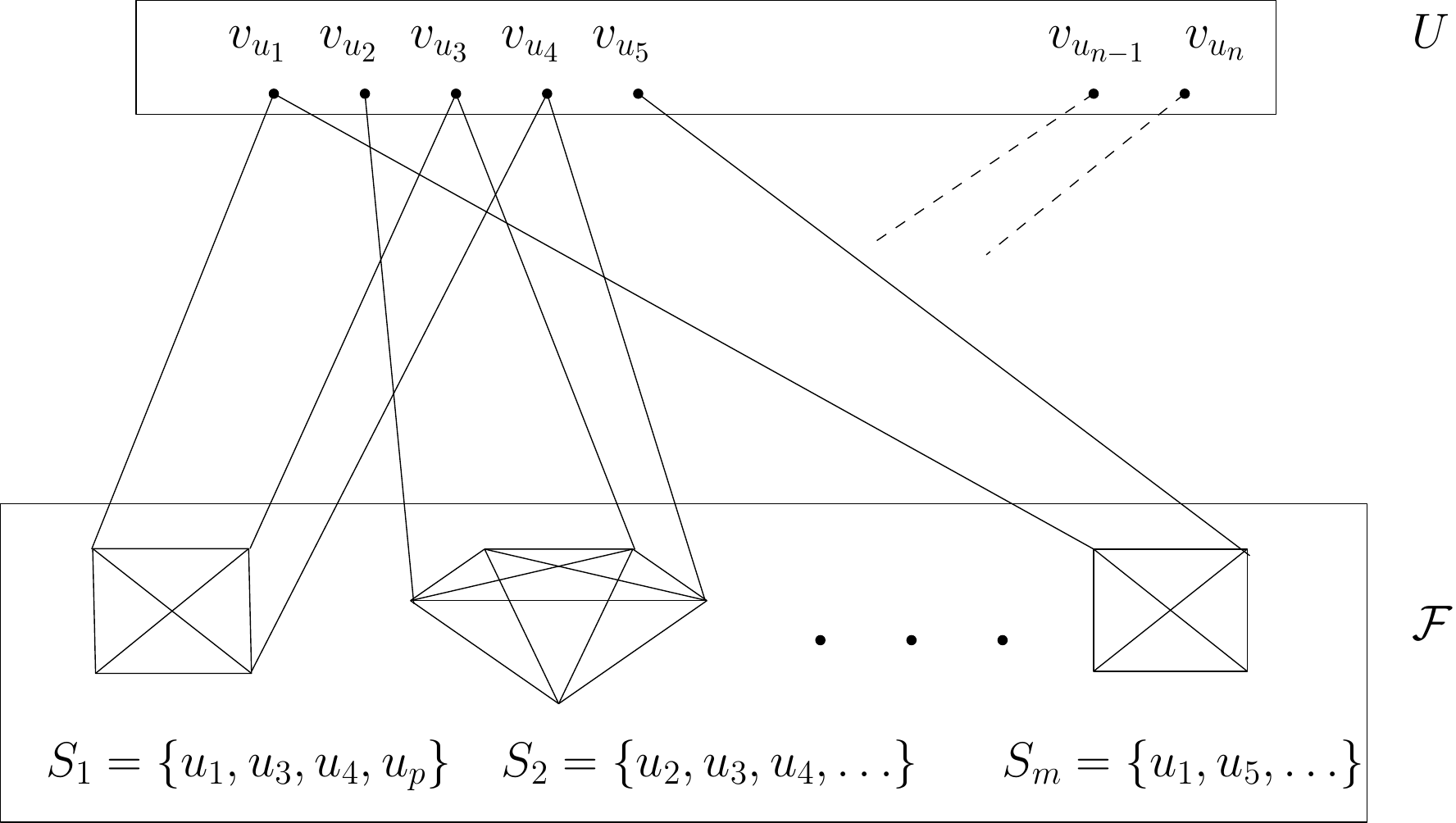}
%\DeclareGraphicsExtensions{.png}
\caption{Reduction from {\sc Hitting Set} to \VCCVD}
\label{figpoint}
\end{center}
\end{figure}

Note that the set of vertices $\bigcup_{u \in U}v_u$ forms a cvd of size $n$ for the graph $G$ we constructed.
  
We claim that there is a hitting set of size $k$ in the instance $(U,\FF)$ if and only if there is a vertex cover of size $k+\sum\limits_{S \in \FF}^{} (|S|-1)$ in $G$.

Let $X \subseteq U$ be the hitting set of size $k$. For each set $S \in \FF$, mark an element of $X$ which intersects $S$. Now we create a subset of vertices $Y$ in $G$ consisting of vertices corresponding to elements in $X$ plus the vertices corresponding to all the unmarked elements in $S$ for every set $S \in \FF$. Clearly $|Y| = k+\sum\limits_{S \in \FF}^{} (|S|-1)$. We claim that $Y$ is a vertex cover of $G$. Let us look at an edge of $G$ between an element vertex $u$ and its corresponding copy vertex in $S$ containing $u$. If $u$ is unmarked in $S$, then it is covered as the vertex corresponding to $u$ in $S$ is present in $Y$. If it is marked, then the element $v_u$ is present in $Y$ which covers the edge. All the other edges of $G$ have both endpoints in a set $S \in \FF$. Since one of them is unmarked, it belongs to $Y$ which covers the edge.

Conversely, let $Z$ be a vertex cover of $G$ of size $k+\sum\limits_{S \in \FF}^{} (|S|-1)$. Since the graph induced on vertices of set $S$ forms a clique for each $S \in \FF$, $Z$ should contain all the vertices of the clique except one to cover all the edges of the clique. Let us mark these vertices. This means that at least $\sum\limits_{S \in \FF}^{} (|S|-1)$ of the vertices of $Z$ are not element vertices $v_u$. Now the remaining $k$ vertices of $Z$ should hit all the remaining edges in $G$. Suppose it contains another vertex $x$ corresponding to an element $u$ in set $S \in \FF$. Since $x$ only can only cover the edge from $x$ to the element vertex $v_u$ out of the remaining edges, we could remove $x$ and add $v_u$ as it is not present in $Z$ and still get a vertex cover of $G$ of the same size. Hence we can assume, without loss of generality that all the remaining vertices of $Z$ are element vertices $v_u$. Let $X'$ be the union of the $k$ elements corresponding to these element vertices. We claim that $X'$ is a hitting set of $(U,\FF)$ of size $k$. Suppose $X'$ does not hit a set $S \in \FF$. Look at the unmarked vertex $x$ in the vertices of $S$. There is an edge from $x$ to its element vertex $v_u$. Since $u \notin X'$, this edge is uncovered in $G$ giving a contradiction.

Hence given a {\sc Hitting Set} instance $(U,\FF)$, we can construct an instance for \VCCVD\ with parameter $n$. Hence, if we could solve \VCCVD\ in $\OO^*((2- \epsilon)^{k})$ time, we can solve {\sc Hitting Set} in $\OO^*((2- \epsilon)^{n})$ time contradicting SETH.
\end{proof}

A graph $G$ is called a cluster graph if it is a disjoint union of complete graphs. We note that in the above reduction, $G \setminus \bigcup_{u \in U}v_u$ forms a cluster graph. Hence we the following corollary.
\begin{corollary}
{\sc Vertex Cover} parameterized by the cluster vertex deletion set size $k$ cannot be solved in $\OO^*((2- \epsilon)^{k})$ time for any $\epsilon >0$ assuming SETH.
\end{corollary}
 \begin{theorem} \label{thm:seth-fvsoct}
{\sc FVS by CVD} and {\sc OCT by CVD} given the modulator cannot be solved in $\OO^*((2- \epsilon)^{k})$ time for any $\epsilon >0$ assuming SETH.
\end{theorem}
\begin{proof}[Proof Sketch]
To prove the above theorem, we again give a reduction very similar to the reduction given in the proof of Theorem \ref{thm:seth-vc}. Consider a {\sc Hitting Set} instance $(U,\FF)$. To create an instance of {\sc Feedback Vertex Set by CVD} or  {\sc Odd Cycle Transversal by CVD}, we replace each edge $e=(u,v)$ in the above reduction by a triangle $t_e$ with vertices $u,v$ and new vertex $v_e$. It can be easily shown that the graph obtained after removing the vertices corresponding to elements in $U$ forms a chordal graph. % We can take all the vertices corresponding to $U$ as a chordal vertex deletion set. 
 %Suppose this set is not a  chordal vertex deletion set. Hence there exists an induced cycle cycle say $v_1,v_2,v_3,\ldots v_p,v_1$ with length at least four. But according to our construction, in the remaining graph either $v_2$ and $v_p$ are connected by an edge, or at least one among $v_2$ and $v_p$ has degree two with both the neighbors connected by an edge, or $ v_2$ and $v_p$ are in different components altogether. This is a contradiction to our assumption. 
% It can be seen that the set taken is indeed chordal vertex deletion set. 
 The proof follows on similar lines.

\end{proof}

\section{Conclusion}

Our main contribution is to develop techniques for addressing structural parameterization problems when the modulator is not given. The question, of Fellows et.al. about whether there is an FPT algorithm for {\sc Vertex Cover} parameterized by perfect deletion set with only a promise on the size of the deletion set, is open.

%We design $2^{\OO(k)}n^{\OO(1)}$ algorithms for {\sc Vertex Cover}, {\sc Feedback Vertex Set} and {\sc Odd Cycle Transversal} where the parameter $k$ is the size of a chordal vertex deletion set in the input graph $G$. 

Regarding problems parameterized by chordal deletion set size,
we remark that not all problems that have FPT algorithms when parameterized by treewidth, necessarily admit an FPT algorithm parameterized by CVD. For example, {\sc Dominating Set} parameterized by treewidth admits an FPT algorithm \cite{cygan2015parameterized} while {\sc Dominating Set} parameterized by CVD is para-NP-hard as the problem is NP-hard in chordal graphs~\cite{Booth1982DominatingSI}. 
% Hence it is unlikely to admit an FPT algorithm unless P=NP.
Generalizing our algorithms for other problems, for example, for the optimization problems considered by Liedloff et al.~\cite{DBLP:journals/algorithmica/LiedloffMT19} to obtain better FPT algorithms
when the modulator is not given, would be an interesting direction.
\subparagraph*{Acknowledgements}
We thank Saket Saurabh for pointing out the known example of % described in the introduction, that are solved without needing the modulators. %the pointing out that the existing algorithm
 {\sc Vertex Cover by Konig Vertex Deletion}  that is solved without needing the modulator as input.
%%
%% Bibliography
%%

%% Please use bibtex, 
\bibliographystyle{abbrv}
\bibliography{main}

%\appendix
%\input{appendix.tex}

\end{document}